\documentclass{article}

\usepackage[english]{babel}
\usepackage[utf8x]{inputenc}
\usepackage[T1]{fontenc}

\usepackage[top=2cm,bottom=3cm,left=2cm,right=2cm,marginparwidth=1.75cm]{geometry}
\usepackage[nottoc]{tocbibind}

\usepackage{amsfonts,amssymb,mathrsfs,amscd,amsmath}
\usepackage{amsthm}
\usepackage{amsfonts}
\usepackage{bbold}
\usepackage{tabularx}
\usepackage{array}

\usepackage{float}
\restylefloat{table}
\numberwithin{equation}{section}
\usepackage{todonotes}

\usepackage{graphicx}
\graphicspath{ {./images/} }
\usepackage[rightcaption]{sidecap}
\usepackage[title]{appendix}
\usepackage{wrapfig}

\usepackage{hyperref}
\hypersetup{
    colorlinks=false,
    pdftitle={Overleaf Example},
    pdfpagemode=FullScreen,
    }

\theoremstyle{plain}
\newtheorem{theorem}{Theorem}
\newtheorem{lemma}{Lemma}
\newtheorem{propos}{Proposition}

\theoremstyle{definition}
\newtheorem{definition}{Definition}

\newtheorem{remark}{Remark}

\newcommand{\Z}{\mathbb{Z}}

\newcolumntype{P}[1]{>{\centering\arraybackslash}p{#1}}
\newcolumntype{L}[1]{>{\raggedright\arraybackslash}p{#1}}
\title{\textbf{Variational formulas on spaces of SL(2) Hitchin's spectral covers}}
\author{\text{Roman Klimov }\footnote{\textit{Email address}: {roman.klimov@concordia.ca} }\\ \\\textit{Department of Mathematics and Statistics, Concordia University}\\ \textit{1455 de Maisonneuve W., Montreal, Quebec, Canada H3G 1M8}}
\date{}

\begin{document}
\maketitle

\begin{abstract}
Using the developed deformation theory on moduli spaces of quadratic differentials we derive variational formulas for objects associated with generalized $SL(2)$ Hitchin’s spectral covers: Prym matrix, Prym bidifferential, Hodge and Prym tau-functions. The resulting formulas are antisymmetric versions of Donagi-Markman residue formula. The second variation of Prym matrix is a natural analogy to the formula previously derived for the period matrix of the $GL(n)$ spectral cover.
\end{abstract}

\tableofcontents{}

\newpage
\section{Introduction}

Hitchin's systems were originally introduced in \cite{hitchin} as a dimensional reduction of the self-dual Yang-Mills equation. These systems together with their meromorphic generalizations \cite{Hurtubise_2000} provide the widest class of integrable systems associated to a Riemann surface. Hamiltonians of such systems are given by meromorphic $l$-differentials arising in the definition of a spectral cover. The moduli spaces of these spectral covers play a particular interest for us in this paper, where we focus on the generalized $SL(2)$ case, naturally identified with the moduli space of meromorphic quadratic differentials on a Riemann surface $\mathcal{C}$. For a given Riemann surface, $Q$  is a family of meromorphic quadratic differentials with generically simple zeroes and poles of arbitrary even orders. Then equation 
$v^2=Q$ in the cotangent bundle of $\mathcal{C}$ defines a two-fold spectral cover $\pi:\hat{\mathcal{C}} \xrightarrow{} \mathcal{C},$ branched at zeroes of $Q$.

Related moduli spaces with a variable base curve were extensively studied. We refer to \cite{Bertola_2017} for their symplectic properties in holomorphic case. Using a natural embedding into the space of Abelian differentials authors derived variational formulas for various
canonical objects associated to the double cover \cite{korotkin2013tau}.  The theory of Bergman tau functions provides an analytical tool to study the rational Picard group of the moduli space of quadratic differentials \cite{korotkin2013tau}, \cite{korotkin2020bergman}.

In this article, following \cite{bertola2019spaces}, we consider variations of objects linked to the covering surface, given a fixed base curve. While in \cite{bertola2019spaces} authors study generalized $GL(n)$ spectral covers, here we focus on the low dimensional case corresponding to the $SL(2)$ gauge group, where all the ramification points have simple branching and the covering surface $\hat{\mathcal{C}}$ possesses an involution automorphism.

The natural involution map $\mu:\hat{\mathcal{C}}\xrightarrow{}\hat{\mathcal{C}}$ on the canonical double cover induces the splitting of its homology group into even and odd subgroups. Similarly, holomorphic differentials defined on the covering surface could be represented as the sum of two differentials symmetric and skew-symmetric under the involution. While the symmetric element is a pullback from the base curve, the skew-symmetric differential is associated exclusively with the covering surface and is called Prym differential. We denote a basis of \textit{normalized Prym differentials} by $u^-_{\alpha}$. The similar decomposition also applies to the canonical bidifferential \cite{korotkin2013tau}. It turns out that only skew-symmetric differentials contribute to the variations under the assumption that the base curve is kept fixed.

Denote by $\{a^-_\gamma, b^-_\gamma \} \in H_{-}(\hat{\mathcal{C}})$ generators of the odd part of homology group of cycles on $\hat{\mathcal{C}}$.  The derivatives of the \textit{Prym matrix} $\Pi_{\alpha \beta}=\oint_{b^-_\beta}u^-_\alpha$ with respect to local coordinates on the space of $SL(2)$ spectral covers reproduce the formula analogous to the Donagi-Markman cubic \cite{Donagi_1996}: denote by $A_\alpha=\oint_{a^{-}_{\alpha}} v$ integrals over $a^-$- cycles, the remaining coordinates defined in (\ref{9}) span a meromorphic part of $v.$

\begin{theorem}
The variation of the Prym matrix $\Pi$ of the covering surface take the following form:

\begin{equation}\label{a1}\frac{\partial \Pi_{\alpha \beta}}{\partial A_{\gamma}}=-\pi i \sum_{x_i}\underset{x_i}{res}\Bigg(\frac{ {u}^-_{\alpha}  {u}^-_{\beta}  {u}^-_{\gamma}}{d{\xi} \: d( v/d{\xi})}\Bigg),\end{equation}

where $\xi$ denotes a local coordinate near a branch point $x_i$ on the base curve $\mathcal{C}$.
\end{theorem}
The variations with respect to the moduli representing the singular part of $v$ are obtained by similar formulas with Prym holomorphic differentials $u^-_{\gamma}$ replaced by  \textit{Prym second-kind (\ref{70}) and third-kind (\ref{71}) differentials} (\mbox{\textbf{Theorem 5}}). The derivation of this result differs from the one proposed in \cite{bertola2019spaces} for $GL(n)$ group where $v$ can be generically viewed as a meromorphic Abelian differential with simple zeroes at $x_i,$ since in $SL(2)$ case differential $v$ gains double zeroes at the branch points.
While in \cite{bertola2019spaces} the contributions to the result bring variations of the form $\underset{x_i}{res}\Big(\frac{u_\alpha u_\beta}{v}\Big)\frac{\partial}{\partial A_\gamma}\Big(\int^{x_i}_{x_1}v \Big)$, in the present context similar contributions arise via residues near the double poles at zeroes of $v.$

Introduce the canonical (Bergman)  bidifferential $\hat{B}(x, y)$ on $\hat{\mathcal{C}} \times \hat{\mathcal{C}}$ and define the \textit{Prym bidifferential} ${B}^-(x, y):=\hat{B}(x, y)-\mu^*_y\hat{B}(x, y),$ where notation $\mu^*_y$
 means that we take the pullback with respect to the involution on the second factor in $\hat{\mathcal{C}} \times \hat{\mathcal{C}}$. The variations of ${B}^-(x, y)$  are given by

\begin{theorem}

\begin{equation}\label{b1}\frac{\partial B^-(x,y)}{\partial A_{\gamma}}=-\frac{1}{2} \sum_{x_i}\underset{x_i}{res}\Bigg(\frac{{u}^-_{\gamma}(t) B^-(x,t)  B^-(t,y)  }{d{\xi} \: d( v/d{\xi})}\Bigg)\end{equation} 

\end{theorem}
and resemble the formulas derived in \cite{Baraglia_2019} and \cite{bertola2019spaces} for $GL(n)$ gauge group.  

While in \cite{bertola2019spaces} variations of the Bergman-tau function were computed in the holomorphic case, in this paper we extend this result. We consider two related Berman tau-functions $\tau^+$ and $\tau^-$, called \textit{Hodge and Prym}, naturally arising in the context of the double cover \cite{Bertola_2020} and derive their variations on the whole space including the derivatives with respect to the moduli encoding the singularities (\textbf{Theorems 8 and 9}). Denote by $\{\tilde{z}_k\}^n_{k=1}$ the set of double poles of $Q,$ then $\pi^{-1}(\tilde{z}_k)=\{\tilde{z}^{(1)}_k, \tilde{z}^{(2)}_k \}$ are corresponding simple zeroes of $v$ with residues denoted by $\tilde{r}_k$ and $-\tilde{r}_k$ respectively.  Derivatives with respect to the periods take the following form

\begin{theorem}

\begin{equation}\label{t+} \frac{\partial \log(\tau^+)}{\partial A_\gamma}= \frac{5}{432}\sum_{x_i}\underset{x_i}{res}\Bigg(\frac{u^-_{\gamma}}{\int^x_{x_i}v} \Bigg)-\sum^n_{k=1}\frac{1}{48\tilde{r}_k}\int^{\tilde{z}^{(1)}_k}_{\tilde{z}^{(2)}_k}u^-_{\gamma},\end{equation}

\begin{equation}\label{t-} \frac{\partial \log(\tau^-)}{\partial A_\gamma}= \frac{1}{2} \sum_{x_i}\underset{x_i}{res}\Bigg(\frac{{u}^-_{\gamma} \hat{B}_{reg}  }{d{\xi} \: d( v/d{\xi})}\Bigg)+\frac{11}{432}\sum_{x_i}\underset{x_i}{res}\Bigg(\frac{u^-_{\gamma}}{\int^x_{x_i}v} \Bigg)-\sum^n_{k=1}\frac{1}{48\tilde{r}_k}\int^{\tilde{z}^{(1)}_k}_{\tilde{z}^{(2)}_k}u^-_{\gamma}.\end{equation} 

\end{theorem}
The term $\hat{B}_{reg}$ appearing in the second formula is the regularization of bidifferential $\hat{B}(x,y)$ near the diagonal:
\begin{equation}\hat{B}_{reg}(x,x)=\Bigg(\hat{B}(x,y)-\frac{v(x)v(y)}{(\int^y_x v)^2} \Bigg)\Bigg|_{y=x}. \end{equation}

The formulas (\ref{148}) (\textbf{Proposition 4}) for the second derivatives of the Prym matrix  
resemble the expressions derived in \cite{bertola2019spaces} using the similar computational approach, or in \cite{Baraglia_2019} by means of topological recursion introduced in \cite{Eynard_2007}. 

This paper is arranged as follows: in \textbf{Section 2} we introduce a coordinate system on the space of generalized $SL(2)$ spectral covers. In \textbf{Section 3} we define a set of coordinates on the space of meromorphic quadratic differentials with variable base, discuss the geometry and main objects associated to the canonical double cover. \mbox{\textbf{Section 4}} is devoted to the definition and main properties of Hodge and Prym tau functions. In \textbf{Section 5} we derive the variational formulas on the space of $SL(2)$ spectral covers with a fixed base.

\textbf{Acknowledgements.} The author thanks his scientific advisor D. Korotkin for posing the problem and fruitful discussions.

\section{Spaces $\mathcal{M}_{SL(2)}[\mathbf{k}]$ of generalized SL(2) spectral covers}

We introduce a Riemann surface $\mathcal{C}$ of genus $g,$ with $m$ marked points $z_1,..., z_m$ and associated positive multiplicities $k_1,...,k_m.$ The Hamiltonians of Hitchin's systems are encoded by meromorphic $l$-differentials on $\mathcal{C}$ arising in the definition of \textit{spectral cover} $\mathcal{\hat{C}}$ given by the following locus in $T^* \mathcal{C}$
\begin{equation}\label{1}\mathcal{\hat{C}}=\{(x,v) \in \mathcal{C} \times T_x^* \mathcal{C}    \ | \ P_n(v)=0 \},\end{equation}
where
\begin{equation}\label{2}P_n(v)=v^n+Q_1 v^{n-1}+...+Q_{n-1}v +Q_n, \end{equation}
$Q_l$ is a meromorphic $l$-differential with poles at $z_j$ of order $lk_j.$ In the framework of \cite{hitchin} the equation (\ref{1}) is given by the characteristic polynomial $P_n(v)=\text{det}(\Phi -vI)$, where $\Phi$ is a \textit{Higgs field} on $\mathcal{C}.$  For the most general case of $GL(n)$ Hitchin's systems the differentials $Q_l$ are arbitrary. This case was recently studied in \cite{bertola2019spaces}. In this paper we focus on the important low rank subgroup $SL(2)$ with
the spectral cover defined by
\begin{equation}\label{3}v^2=Q,\end{equation}
with $Q$ being a quadratic differential with simple zeroes and poles at $z_j$ of order $2k_j$ thanks
to the genericity assumption. For a fixed base $\mathcal{C}$ and positions of poles we introduce the moduli space $\mathcal{M}_{SL(2)}[\mathbf{k}]$ of quadratic differentials with simple zeroes and poles of associate orders  $\mathbf{k}=(2k_1,...,2k_m).$ The degree $4g-4$ of the divisor class $(Q)$ implies a quadratic differential has \begin{equation}\label{4}r=4g-4+2\sum^m_{j=1}k_j\end{equation} simple zeroes denoted by $x_i$. For all such quadratic differentials the equation (\ref{3}) in the cotangent bundle $T^*\mathcal{C}$ defines double covering 
$\pi:\hat{\mathcal{C}}\xrightarrow{}\mathcal{C},$ branched at zeroes of $Q.$  The covering surface $\hat{\mathcal{C}}$ possesses a natural holomorphic involution $\mu:\hat{\mathcal{C}} \xrightarrow{}\hat{\mathcal{C}}.$ $v="\sqrt{Q}"$ is a single-valued meromorphic Abelian differential on $\hat{\mathcal{C}}$ and skew-symmetric under the involution: $v(x^{\mu})=-v(x)$. 

Differential $v$ has double zeroes at branch points $x_i$  which follows from the following short observation: given $\xi$ is a local coordinate near any zero $x_i$ on $\mathcal{C}$ s.t. $Q=\xi (d\xi)^2$, one has $v=\sqrt{Q}=\sqrt{\xi} d\xi =2 \hat{\xi}^2d\hat{\xi}$, where $\hat{\xi}=\sqrt{\xi}$ is a local coordinate near $x_i$ on $\hat{\mathcal{C}}$. Then the Riemann-Hurwitz formula implies the genus of the covering surface $\hat{\mathcal{C}}$ equals
\begin{equation}\label{5}\hat{g}=4g-3+\sum^m_{j=1} k_j. \end{equation}

Since the branch points of $\hat{\mathcal{C}}$ do not coincide with $z_j$ we have $\pi^{-1}(z_j)=\{z^{(1)}_j, z^{(2)}_j \}.$
The differential $v$ has on $\hat{\mathcal{C}}$ poles of order $k_j$ at both $z^{(1)}_j$ and $z^{(2)}_j.$ Denote by $\chi_j$ a local coordinate on $\mathcal{C}$ near $z_j.$ We can also use $\chi_j$ as local coordinate near both $z^{(1)}_j$ and $z^{(2)}_j.$  Consider
the singular parts of $v$ near $z^{(1)}_j:$

 \begin{equation}\label{6}v(\chi_j(x))\big|_{x \xrightarrow{} z^{(1)}_j}=\Bigg(\frac{C^{k_j}_j}{\chi^{k_j}_j}+\frac{C^{k_j-1}_j}{\chi^{k_j-1}_j}+...+\frac{C^{1}_j}{\chi_j}+O(1)\Bigg) d \chi_j.\end{equation} As $v$ is skew-symmetric under the involution and $\mu(z^{(1)}_j)=z^{(2)}_j,$ we have that near a point $z^{(2)}_j$ using the same coordinate $\chi_j$ we have the following expansion:

\begin{equation}\label{7}v(\chi_j(x))\Big|_{x \xrightarrow{} z^{(2)}_j}=\Bigg(-\frac{C^{k_j}_j}{\chi^{k_j}_j}-\frac{C^{k_j-1}_j}{\chi^{k_j-1}_j}-...-\frac{C^{1}_j}{\chi_j}+O(1)\Bigg) d \chi_j.\end{equation}

The dimension of $\mathcal{M}_{SL(2)}[\mathbf{k}]$ consists of the sums of dimensions of meromorphic and holomorphic parts of a quadratic differential $Q$ which equals
\begin{equation}\label{8}\text{dim}\mathcal{M}_{SL(2)}[\mathbf{k}]=2\sum^{m}_{j=1} k_j + (3g-3). \end{equation}
We introduce the following set of local coordinates on the moduli space $\mathcal{M}_{SL(2)}[\mathbf{k}]$:
\begin{equation}\label{9}\Big\{\{A_\alpha\}^{3g-3+\sum^m_{j=1} k_j}_{\alpha=1}, \{C^{l}_j \}, \ j=1,...,m, \  l=1,...,k_j \Big\}. \end{equation}
while $C^{l}_j$ are coefficients of singular parts of $v$ near $z^{(1)}_j$ (or $z^{(2)}_j$), $A_\alpha$ are integrals over skew-symmetric part of the $a$ - cycles on $\hat{\mathcal{C}}$  (defined explicitly in the following section)
\begin{equation}\label{10}A_\alpha=\oint_{a^{-}_{\alpha}} v.\end{equation}
\begin{remark}

While $C^{1}_j$ are coordinate-independent residues, $C^{l}_j, \ l \geq 2$ clearly depend on the choice of local coordinates $\chi_j.$ They are familiar from the theory of
algebro-geometric solutions of Kadomtsev-Petviashvili (KP) equation where they are called the KP-times. Integrals $A_\alpha$ depend on the choice of Torelli marking on the base surface $\mathcal{C}$. 
\end{remark}
\section{Spaces $\mathcal{Q}_{g,m}[\mathbf{k}]$ of meromorphic quadratic differentials}
\subsection{Geometry of double cover}
Denote by $\mathcal{Q}_{g,m}[\mathbf{k}]$ moduli space of pairs: a Riemann surface $\mathcal{C}$ of genus $g$ and a meromorphic quadratic differential $Q$ with poles at $\{z_j\}^m_{j=1}$ of corresponding orders  $\textbf{k}=\{2k_j\}^m_{j=1}$ and $4g-4+2\sum^m_{j=1}k_j$ simple zeroes. 
The dimension of $\mathcal{Q}_{g,m}[\mathbf{k}]=(\mathcal{C},Q)$ consists of $3g-3$ modular parameters of $\mathcal{C},$ $m$ positions of singularities, $2\sum^m_{j=1}k_j$ coefficients of singular parts and $3g-3$ parameters that form a holomorphic part of $Q.$ Thus, the total dimension is
\begin{equation}\label{11}\text{dim}\mathcal{Q}_{g,m}[\mathbf{k}]=6g-6+m+2\sum^m_{j=1}k_j. \end{equation}
The definition of covering surface $\hat{\mathcal{C}}$ of genus $\hat{g},$ projection $\pi$ and involution $\mu$ are in accordance with the previous section.
We decompose the first homology group of $H_1(\hat{\mathcal{C}}\backslash \{z^{(1)}_j, z^{(2)}_j \}^m_{j=1} )$ into
 \begin{equation}\label{12}H_1(\hat{\mathcal{C}}\backslash \{z^{(1)}_j, z^{(2)}_j \}^m_{j=1} )=H_+ \oplus H_- ,\end{equation}
which are the $+1$ and $-1$ eigenspaces of the map, induced by the involution $\mu.$ $\text{dim}(H_+)=2g+m-1$ and $\text{dim}(H_-)=6g-6+2\sum^m_{j=1}k_j+m. $
The canonical basis of $H_1(\hat{\mathcal{C}}\backslash \{z^{(1)}_j, z^{(2)}_j \}^n_{j=1} )$ can be chosen as follows:
\begin{equation}\label{13}\{a_k, a^\mu_k, \tilde{a}_l, b_k, b^\mu_k, \tilde{b}_l, t_j, t^\mu_j  \}, \quad k=1,...,g, \quad  l=1,...,2g-3+\sum^m_{j=1}k_j, \quad j=1,...,m.\end{equation}
Here
$\{ a_k, b_k, a^\mu_k, b^\mu_k \} $ is a lift of the canonical basis of cycles $\{ a_k ,b_k \}$ from $\mathcal{C}$ to $\hat{\mathcal{C}}$ such that
\begin{equation}\label{14}\mu_*a_k=a^\mu_k, \quad \mu_*b_k=b^\mu_k, \quad \mu_*\tilde{a}_l+\tilde{a}_l=\mu_*\tilde{b}_l+\tilde{b}_l=0. \end{equation}
$\{t_j, t_j^\mu \}$ is a lift of a small positively-oriented loop $t_j$ around $z_j$ on $\mathcal{C}.$ On double cover $\hat{\mathcal{C}},$  $t_j$ denotes a positively-oriented loop encircling $z^{(1)}_j,$ while $t_j^\mu$ is a small loop around $z^{(2)}_j$. In the group (\ref{12}) there is a single relation given by
\begin{equation}\label{15}\sum^{m}_{i=1}(t_j+t^\mu_j)=0. \end{equation}
The classes 
\begin{equation}\label{16}a^+_k=\frac{1}{2}(a_k+a^\mu_k), \quad b^+_k=\frac{1}{2}(b_k+b^\mu_k), \quad t^+_j=\frac{1}{2}(t_j+t^\mu_j)\end{equation} generate the group $H^+$ with the intersection index
\begin{equation}\label{17}a^+_i \circ b^+_k=\frac{1}{2}\delta_{ik}, \end{equation}
while $t^+_j$'s have zero intersection with all cycles. 
The following cycles
\begin{equation}\label{18}a^-_k=\frac{1}{2}(a_k-a^\mu_k), \quad b^-_k=\frac{1}{2}(b_k-b^\mu_k),\end{equation}
\begin{equation}\label{19}a^-_l=\frac{1}{\sqrt{2}}\tilde{a}_l, \quad b^-_l=\frac{1}{\sqrt{2}}\tilde{b}_l,\end{equation}

\begin{equation}\label{20}t^-_j=\frac{1}{2}(t_j-t^\mu_j) \end{equation}
are the generators of the group $H_-.$ Similarly, their intersection index is 
\begin{equation}\label{21}a^-_i \circ b^-_k=\frac{1}{2}\delta_{ik} \end{equation}
and all other intersections are zero.

The dimension of $H_-$ coincides with the dimension of $\mathcal{Q}_{g,m}[\mathbf{k}]$. We introduce the following set of \textit{period} (homological) local coordinates on $\mathcal{Q}_{g,m}[\mathbf{k}]$:
\begin{equation}\label{22}A_k=\oint_{a^-_k}v, \quad B_k=\oint_{b^-_k}v, \quad 2 \pi r_j= \oint_{t^-_j}v, \end{equation}
here $+ r_j$ and $- r_j$ are residues of $v$ near $z^{(1)}_j$  and $z^{(2)}_j$, respectively.

This coordinate system is a further generalization of period coordinates used to study moduli spaces of quadratic differentials in holomorphic case \cite{Bertola_2017}, and in cases of presence of simple \cite{korotkin2018periods} and  double poles \cite{Bertola_2020}, \cite{bertola2021wkb} .

\subsection{Standard meromorphic objects}
In this section we introduce basic meromorphic objects associated with the canonical double cover. The involution map $\mu:\hat{\mathcal{C}} \xrightarrow{} \hat{\mathcal{C}} $ yields a decomposition of the first cohomology group into even and odd parts \begin{equation}\label{23}H^{1,0}(\hat{\mathcal{C}})=H^{+}(\hat{\mathcal{C}})\oplus H^{-}(\hat{\mathcal{C}}).\end{equation} 
We shall denote by  \begin{equation}\label{24}\{\hat{u}_k, \hat{u}^{\mu}_k, \hat{w}_l \}, \quad k=1,...,g, \qquad l=g+1, ..., \hat{g}.\end{equation}  the basis of normalized holomorphic Abelian differentials on $\hat{\mathcal{C}}$ dual to the basis of cycles (\ref{13}). The differentials $u^+_k=\hat{u}_k+\hat{u}^{\mu}_k$ provide a basis in $H^{+}(\hat{\mathcal{C}}).$ These differentials are invariant under involution and naturally isomorphic to the space of holomorphic differentials on $\mathcal{C}$; $\dim(H^{+}(\hat{\mathcal{C}}))=g$. The space $H^{-}(\hat{\mathcal{C}})$ consists of holomorphic differentials on $\hat{\mathcal{C}}$ with a skew-symmetric property $u(x^{\mu})=-u(x).$ Such elements are called \textit{Prym holomorphic differentials}. The basis for $H^{-}(\hat{\mathcal{C}})$ is generated by
\begin{equation}\label{25}
u^-_l=
\begin{cases}
  \hat{u}_l-\hat{u}^{\mu}_l, \quad l=1,...,g,\\    
  \sqrt{2} \hat{w}_l, \quad l=g+1, ..., \hat{g}.   
\end{cases}
\end{equation}
Notice, that differentials $u^-_l$ are normalized over $a^-$- cycles and have vanishing $a^+$- periods. \begin{equation}\label{26}\dim(H^{-}(\hat{\mathcal{C}}))=3g-3+\sum^m_{j=1} k_j:=g^-.\end{equation}

Integrating these differentials over corresponding even and odd parts of $b$ - cycles we obtain Period and \textit{Prym matrices}
\begin{equation}\label{27}\Omega_{ij}=\oint_{b^+_j}u^+_i, \quad \Pi_{ij}=\oint_{b^-_j}u^-_i. \end{equation}
$\Omega$ is identified with the period matrix of the base surface $\mathcal{C}.$ The period matrix $\hat{\Omega}$ of the cover $\hat{\mathcal{C}}$ could be expressed via $\Omega$ and $\Pi$ applying an appropriate linear transformation \cite{Bertola_2020}.

We proceed with bidifferentials and projective connections on double covers.
Let $\hat{B}(x, y)$ denote the canonical (Bergman) bidifferential on $\hat{\mathcal{C}} \times \hat{\mathcal{C}}$ associated with
the homology basis (\ref{13}).  $\hat{B}(x, y)$  is symmetric, has the second order pole on the
diagonal $x = y$ with biresidue 1 and satisfies

\begin{equation}\label{28}\oint_{a_k} \hat{B}(\cdot, y) = \oint_{a^\mu_k} \hat{B}(\cdot, y) = \oint_{\tilde{a}_l} \hat{B}(\cdot, y)=0.\end{equation}
Equivalently, that can be written in terms of the $a^+$ and $a^-$- cycles:

\begin{equation}\label{29}\oint_{a^+_k} \hat{B}(\cdot, y) = \oint_{a^-_k} \hat{B}(\cdot, y)=0.\end{equation}
We put
\begin{equation}\label{30}{B}^+(x, y):=\hat{B}(x, y)+\mu^*_y\hat{B}(x, y), \end{equation}
\begin{equation}\label{31}{B}^-(x, y):=\hat{B}(x, y)-\mu^*_y\hat{B}(x, y).\end{equation}
(notation $\mu^*_y$
 means that we take the pullback with respect to the involution on the second factor in $\hat{\mathcal{C}} \times \hat{\mathcal{C}}$). ${B}^+(x, y)$ is the pullback of the
canonical bidifferential $B(x, y)$ on ${\mathcal{C}} \times {\mathcal{C}}$ (normalized relative to the $a$ - cycles on the base),
the bidifferential $B^{−}(x,y)$ is called the \textit{Prym bidifferential} \cite{korotkin2013tau}.
It follows from the definitions that $B^+(x,y)$ and $B^-(x,y)$  are symmetric and skew-symmetric under the involution in both arguments, respectively.

Near the diagonal $x=y$ on $\hat{\mathcal{C}} \times \hat{\mathcal{C}}$ we have

\begin{equation}\label{32}\hat{B}(x,y)=\Big(\frac{1}{(\xi(x)-\xi(y))^2}+\frac{1}{6}\hat{S}_B(\xi(x))+...\Big)d\xi(x)d\xi(y), \end{equation}
as $y \xrightarrow{} x$ for any local coordinate $\xi$ on $\hat{\mathcal{C}}.$ The term $\hat{S}_B$ transforms like a projective connection under the change of coordinates. It is called the Bergman projective connection. For $B^{\pm}(x, y)$ near the diagonal we have
\begin{equation}\label{33}{B}^{\pm}(x,y)=\Big(\frac{1}{(\xi(x)-\xi(y))^2}+\frac{1}{6}{S}^{\pm}_B(\xi(x))+...\Big)d\xi(x)d\xi(y), \end{equation}
with two projective connections $S^+_B$ and $S^-_B$ that are related by
\begin{equation}\label{34}S^{\pm}_B(x)=\hat{S}_B(x) \pm 6\mu^*_y\hat{B}(x, y)|_{y \xrightarrow{} x}. \end{equation}
We call $S^-_B$ the \textit{Prym projective connection}. Note that while $\hat{S}_B$ is holomorphic on $\hat{\mathcal{C}},$ $S^{\pm}_B$ have poles at branch points.

\subsection{Variational formulas on $\mathcal{Q}_{g,m}[\mathbf{k}]$}
 If $x$ is a point on $\hat{\mathcal{C}}$ which does not coincide with branch points $\{x_i\}$ and poles $\{ z^{(1)}_j, z^{(2)}_j \},$ then the local coordinate (also called "flat" coordinate) near $x$ could be taken as
\begin{equation}\label{35}z(x)=\int^x_{x_1}v,\end{equation}
where $x_1$ is a chosen "first" zero of $v$ (notice that in this case $v(x)=dz(x)$).  $z(x)$ could also be used as a coordinate on $\mathcal{C}$ outside branch points and poles. 

It is convenient to introduce the periods $\mathcal{P}_{s_i}=\oint_{s_i}v$ for $s_i$ being an element from the canonical basis of $H_{-}:$
\begin{equation}\label{36}\{s_i\}^{\text{dim}(H_-)}_{i=1}= \Big\{\{a^-_k, b^-_k \}^{g^-}_{k=1}, \{t^-_j \}^{m}_{j=1} \Big\}. \end{equation}
The dual basis $\{s^*_i\}$ is defined by the condition
\begin{equation}\label{37}s^*_i \circ s_j = \delta_{ij} \end{equation}
and is given by
\begin{equation}\label{38}\{s^*_i\}^{\text{dim}(H_-)}_{i=1}= \Big\{\{-2b^-_k, 2a^-_k \}^{g^-}_{k=1}, \{2\kappa^-_j \}^{m}_{j=1} \Big\}, \end{equation}
here $\kappa^-_j$ is a $1/2$ of the contour connecting poles $z^{(1)}_j$ with $z^{(2)}_j$ and skew-symmetric under the involution, not intersecting other contours. Such generator may be chosen as follows: connect $z^{(1)}_j$ with a branch point $x_j$ by an arc $l_j$ on a first copy of $\mathcal{C}$. Then join $x_j$ with $z^{(2)}_j$ on second copy of $\mathcal{C}$ by the antisymmetrized arc $\mu(l_j)$. 

The following variational formulas were derived in \cite{Bertola_2017} for a holomorphic differential $v$. In our framework for $v$ being meromorphic the same formulas apply since the proof does not rely on the presence of poles of $v$. Note that the variations of the differentials depending on the point $x \in \hat{\mathcal{C}}$ are computed assuming that the coordinate $z(x)$ is independent of the moduli. 

\begin{propos}

For a basis $\{s_i\}^{dim(H_{-})}_{i=1}$ of $H_{-}(\hat{\mathcal{C}}\backslash \{z^{(1)}_j, z^{(2)}_j \}^m_{i=1}) $ and its dual basis $\{ s^{*}_{i} \}^{dim(H_{-})}_{i=1}$ the following formulas hold on  $\mathcal{Q}_{g,m}[\mathbf{k}]$:
\begin{equation}\label{39}\frac{\partial \Omega_{ij}}{\partial \mathcal{P}_s}=\frac{1}{2}\oint_{s^*}\frac{u^+_i u^+_j}{v},\end{equation}
\begin{equation}\label{40}\frac{\partial \Pi_{ij}}{\partial \mathcal{P}_s}=\frac{1}{2}\oint_{s^*}\frac{u^-_i u^-_j}{v},\end{equation}
\begin{equation}\label{41}\frac{\partial u^{\pm}_j(x)} {\partial \mathcal{P}_s}\Big|_{z(x)=const}=\frac{1}{4 \pi i}\oint_{s^*}\frac{u_j^{\pm}(t) B^{\pm}(x,t)}{v(t)},\end{equation}

\begin{equation}\label{42}\frac{\partial B^{\pm}(x,y)}{\partial \mathcal{P}_s}\Big|_{z(x),z(y)=const}=\frac{1}{4 \pi i}\oint_{s^*}\frac{B^{\pm}(x,t)B^{\pm}(t,y)}{v(t)}.\end{equation}
\end{propos}

\section{Hodge and Prym tau-functions on $\mathcal{Q}_{g,m}[\mathbf{k}]$} 
The Bergman tau-function on moduli spaces of differentials was originally defined as a higher genus generalization of the Dedekind eta function on elliptic surface. Beginning from the moduli space of holomorphic Abelian differentials \cite{Kokotov_2009} it was extended to the case of Abelian differential with arbitrary divisor \cite{Kalla_2014}; further generalizations cover moduli spaces of holomorphic quadratic \cite{Bertola_2017} and N-differentials \cite{korotkin2017tau}. Recently in \cite{baker2020class} Bergman tau-function was applied to describe the discriminant class of $Sp(2n)$ Hitchin’s spectral covers. In our framework we consider a moduli space of quadratic meromorphic differentials with simple zeroes. 
With this space we associate two naturally arising Bergman tau-functions: Hodge and Prym. The first tau-function is a holomorphic section of the determinant line bundle of the Hodge
vector bundle; the second tau-function is a section of the determinant of the Prym vector bundle, hence their names. While in \cite{Bertola_2020}, \cite{bertola2021wkb} the presence of double poles was considered, here we allow to have poles of arbitrary even order.

For the purpose of the explicit definition of tau-functions we introduce a special system of local coordinates on the double cover called $\textit{distinguished}$. It differs from that used in (\ref{6}) to code a singular part of the differential $v.$ While the first one will be used to define tau functions on $\mathcal{Q}_{g,m}[\mathbf{k}],$ the latter appears when deriving variations of the tau functions on  $\mathcal{M}_{SL(2)}[\mathbf{k}]$ in Section 5.

\subsection{Distinguished local coordinates on $\mathcal{C}$ and $\hat{\mathcal{C}}$}

The quadratic differential $Q$ on the base curve $\mathcal{C}$ allows us to define the set of distinguished local coordinates on both surfaces $\mathcal{C}$ and $\hat{\mathcal{C}}.$ Denote by $\{\tilde{z}_k\}^n_{k=1} \subset \{{z}_j\}^m_{j=1}  $ a subset of poles of of order 2. Then the divisor of $Q$ looks as follows:
\begin{equation}\label{43}(Q)=\sum^{r+m}_{i=1}{d}_i {q}_i\equiv\sum^r_{i=1} x_i - \sum^n_{j=1} 2 \tilde{z}_j- \sum^{m-n}_{j=1} 2k_j z_j, \ k_j \geq 2.\end{equation} The divisor of Abelian differential $v$ on $\hat{\mathcal{C}}$ is given by
\begin{equation}\label{44}(v)=\sum^{r+2m}_{i=1}\hat{d}_i \hat{q}_i\equiv \sum^r_{i=1} 2x_i - \sum^{n}_{j=1} (\tilde{z}^{(1)}_j+\tilde{z}^{(2)}_j) - \sum^{m-n}_{j=1} k_j (z^{(1)}_j+z^{(2)}_j), \ k_j \geq 2. \end{equation}

$\bullet$ Near any point $x_0 \in \hat{\mathcal{C}}$ such that $\pi(x_0) \not \in (Q)$ the local coordinates on ${\mathcal{C}}$ and $\hat{\mathcal{C}}$ can be chosen as $z(x)=\int^{x}_{x_0}v.$

$\bullet$ Near a branch point $x_i$ local parameters $\hat{\zeta}_i$ on $\hat{\mathcal{C}}$ and ${\zeta}_i$ on ${\mathcal{C}}$ are given by
\begin{equation}\label{45}\hat{\zeta}_i(x)=\Bigg(\int^x_{x_i}v\Bigg)^{\frac{1}{3}}, \quad {\zeta}_i(x)=\hat{\zeta}^2_i(x)=\Bigg(\int^x_{x_i}v\Bigg)^{\frac{2}{3}}. \end{equation}

$\bullet$ In the neighborhood of a double pole $\tilde{z}_j$ on $\mathcal{C}$ and corresponding simple poles $(\tilde{z}^{(1)}_j, \tilde{z}^{(2)}_j)$ on $\hat{\mathcal{C}}$ the local coordinate is
\begin{equation}\label{46}\zeta_j(x)=\exp \Big(\frac{1}{ \tilde{r}_j} \int^x_{x_1}v\Big), \end{equation}
where $x_1$ is a chosen first zero of $v;$ $\ \tilde{r}_j$ is a residue of $v$ on $\hat{\mathcal{C}}$ defined in (\ref{22}).

\begin{remark}
Definitions of tau functions depend on the choice of local coordinates near poles $\tilde{z}_j$. To define these coordinates uniquely we on $\mathcal{C}$ connect first zero $x_1$ with a chosen first double pole $\tilde{z}_1$ by a branch cut $\gamma_1$, then connect $\tilde{z}_1$ with the remaining double poles $\{\tilde{z}_j\}^n_{j=2}$ by $\gamma_j$ forming a tree graph $G$. Then we lift $G$ to $\hat{\mathcal{C}}$ via $\pi^{-1}$ and denote the corresponding lift by $\hat{G}=\pi^{-1}(G).$
\end{remark}

$\bullet$ If $k_j \geq 2$ one has a pole of order $2k_j$ at $z_j$ on $\mathcal{C}$ and corresponding poles $(z^{(1)}_j, z^{(2)}_j)$ of order $k_j$   on $\hat{\mathcal{C}}$ with nontrivial residues $\pm r_j.$ The local coordinate on both $\mathcal{C}$ and $\hat{\mathcal{C}}$ in this case is defined from the following transcendental equations:
\begin{equation}\label{47}v=\Bigg(\frac{1-k_j}{\zeta^{k_j}_j}+ \frac{r_j}{\zeta_j} \Bigg) d \zeta_j\end{equation} 
or
\begin{equation}\label{48}\frac{1}{\zeta^{k_j-1}_j}  +r_j \ln \zeta_j= \int^x_{x_1}v. \end{equation}

\subsection{Definition and properties of $\tau^+$ and $\tau^-$}
Denote by $E(x,y)$ the Prime form on $\mathcal{C},$ by $\mathcal{A}_x$ the Abel map with $x$ as a base point and by $K^x$ the vector of Riemann constants. Introduce two vectors $\textbf{r}, \textbf{s} \in \frac{1}{2}\Z^g$ such that
\begin{equation}\label{49}\frac{1}{2}\mathcal{A}_x((Q))+2K^x+\Omega \textbf{r}+ \textbf{s}=0\end{equation}
and the following notations:

\begin{equation}\label{50}E(x,q_i)=\lim_{y \xrightarrow{}q_i}E(x,y)\sqrt{d\zeta_i(y)}, \end{equation}
\begin{equation}\label{51}E(q_i,q_j)=\lim_{x \xrightarrow{} q_i,y \xrightarrow{}q_j}E(x,y)\sqrt{d\zeta_i(x)}\sqrt{d\zeta_j(y)}, \end{equation}
where $\zeta_i$ is the distinguished local parameter on $\mathcal{C}$ near a point $q_i$ from the list (\ref{43}).

Consider the following   multi-valued $g(1-g)/2$ - differential $C(x)$ on $\mathcal{C}$

\begin{equation}\label{52}C(x)=\frac{1}{W(x)}\Bigg(\sum^g_{i=1} u_i(x) \frac{\partial}{\partial w_i} \Bigg)^g \theta(w, \Omega) \Big|_{w=K^x}, \quad W(x):= det \Bigg[ \frac{d^{k-1}}{dx^{k-1}}u_j\Bigg]_{1 \leq j, k \leq g}\end{equation}
here $\Omega$ is the period matrix of the base curve $\mathcal{C}$, $\{u_j\}^g_{j=1}$ are normalized holomorphic differentials on $\mathcal{C}$ and  $\theta$ is the corresponding theta-function.

\begin{definition}
For a given choice of Torelli marking and tree graph $G$ on $\mathcal{C}$ \textit{the Hodge tau-function} $\tau^+$ is given by the following expression
:

\begin{equation}\label{53}\tau^+(\mathcal{C},Q)=C^{2/3}(x)\Bigg(\frac{Q(x)}{\prod^{r+m}_{i=1}E^{d_i}(x,q_i)} \Bigg)^{(g-1)/6}\prod_{i<j}E(q_i,q_j)^{d_i d_j/24} \ \ e^{\frac{- \pi i}{6}<\Omega \textbf{r},\textbf{s}>-\frac{2 \pi i}{3}<\textbf{r},K^x>}.\end{equation}
\end{definition}

Similarly, we denote by $\hat{E}(x,y)$ and $\hat{C}(x)$ the Prime form and multi-valued differential, by $\hat{\mathcal{A}}_x$ the Abel map and by $\hat{K}^x$ the vector of Riemann constants on the double cover $\hat{\mathcal{C}}$. Introduce two vectors $\hat{\textbf{r}}, \hat{\textbf{s}} \in \Z^{\hat{g}}$ such that
\begin{equation}\label{54}\hat{\mathcal{A}}_x((v))+2\hat{K}^x+\hat{\Omega} \hat{\textbf{r}}+ \hat{\textbf{s}}=0\end{equation}
and the notations:

\begin{equation}\label{55}\hat{E}(x,\hat{q}_i)=\lim_{y \xrightarrow{}\hat{q}_i}\hat{E}(x,y)\sqrt{d\hat{\zeta}_i(y)}, \end{equation}
\begin{equation}\label{56}\hat{E}(\hat{q}_i,\hat{q}_j)=\lim_{x \xrightarrow{} \hat{q}_i,y \xrightarrow{}\hat{q}_j}\hat{E}(x,y)\sqrt{d\hat{\zeta}_i(x)}\sqrt{d\hat{\zeta}_j(y)}, \end{equation}

where $\hat{\zeta}_i$ is the distinguished local parameter on $\hat{\mathcal{C}}$ near $\hat{q}_i$ from the list (\ref{44}).

\begin{definition}
For a given choice of Torelli marking and tree graph $\hat{G}$ on $\hat{\mathcal{C}}$ the tau-function $\hat{\tau}$ is given by the following formula:

\begin{equation}\label{57}\hat{\tau}(\hat{\mathcal{C}},v)=\hat{C}^{2/3}(x)\Bigg(\frac{v(x)}{\prod^{{r+2m}}_{i=1}\hat{E}^{\hat{d}_i}(x,\hat{q}_i)} \Bigg)^{(\hat{g}-1)/3}\prod_{i<j}\hat{E}(\hat{q}_i,\hat{q}_j)^{\hat{d}_i \hat{d}_j/6} \ \ e^{-\frac{ \pi i}{6}<\hat{\Omega} \hat{\textbf{r}},\hat{\textbf{s}}>-\frac{2 \pi i}{3}<\hat{\textbf{r}},\hat{K}^x>}.\end{equation}
\end{definition}
\begin{definition}
\textit{The Prym tau-function} is defined by

\begin{equation}\label{58}\tau^-:=\frac{\hat{\tau}(\hat{\mathcal{C}},v)}{\tau^+(\mathcal{C},Q)}. \end{equation}
\end{definition}

The following properties for $\tau^{\pm}$ generalize the ones in \cite{bertola2021wkb} where only double poles of $Q$ were considered:

$\bullet$ The expressions (\ref{53}) and (\ref{58}) for $\tau^{\pm}$ do not depend on the point $x$ although it seems that they do \cite{Kokotov_2009}.

$\bullet$ Let the matrices 
$\begin{pmatrix}A^+ & B^+ \\ C^+ & D^+  \end{pmatrix},  \begin{pmatrix}A^- & B^- \\ C^- & D^-  \end{pmatrix}$ $\in Sp(2g, \Z)$ denote symplectic transformations of the Torelli marking in $H_{+}(\hat{\mathcal{C}})$ and $H_{-}(\hat{\mathcal{C}})$, respectively. Then the tau functions transform as  

\begin{equation}\label{59}\tau^+ \xrightarrow{}  \epsilon \ \text{det}(C^+ \Omega +D^+) \ \tau^+, \qquad \tau^- \xrightarrow{}  \epsilon \ \text{det}(C^- \Pi +D^-) \ \tau^-,\end{equation}
where $\epsilon^{48}=1$; $\Omega, \Pi$ are Period and Prym matrices defined by (\ref{27}) \cite{Bertola_2020}, \cite{korotkin2013tau}.

$\bullet$ The expressions 
\begin{equation}\label{60}(\tau^+)^{48} \prod^n_{k=1} (d \zeta_k (\tilde{z}_k) )^4, \qquad (\tau^-)^{48} \prod^n_{k=1} (d \zeta_k (\tilde{z}_k) )^4  \end{equation} are invariant under the choice of local parameters $\zeta_k$ near $\tilde{z}_k$ \cite{Bertola_2020}.

$\bullet$ The functions $\tau^{\pm}$ satisfy the following homogeneity properties \cite{Kalla_2014}, \cite{korotkin2020bergman}:
\begin{equation}\label{61}\tau^+(\mathcal{C},\epsilon Q)= \epsilon^{\kappa^{+}}\tau^+(\mathcal{C},Q), \quad  \kappa^+=\frac{1}{48}\sum_{d_i \neq -2}\frac{d_i(d_i+4)}{d_i+2},  \end{equation}
\begin{equation}\label{62}\tau^-(\mathcal{C},\epsilon Q)= \epsilon^{\kappa^{-}}\tau^-(\mathcal{C}, Q), \quad  \kappa^-=\frac{1}{24}\sum_{\hat{d}_i \neq -1}\frac{\hat{d}_i(\hat{d}_i+2)}{\hat{d}_i+1}- \kappa^+. \end{equation}

$\bullet$ The expressions for $\tau^{\pm}$ depend on the choice of the first zero $x_1$ and on the integration paths between $x_1$ and poles $\tilde{z}_i$
which are chosen in the complement of the tree graph $\hat{G}$;  the change of the graph $\hat{G}$ within the fundamental polygon $\hat{\mathcal{C}}_0$ affect the coordinates $\zeta_i$ 
near  $\tilde{z}_i$ by a factor of the form

\begin{equation}\label{63}\exp{ \Bigg\{2 \pi i\sum_{j,k}n_{jk}\frac{\tilde{r}_j}{\tilde{r}_k} \Bigg\}}, \end{equation}
where $n_{jk}$ is a matrix of integers \cite{Kalla_2014}.
\subsection{Differential equations for $\tau^{+}$ and  $\tau^{-}$   }
Consider the regularization of Bergman and Prym bidifferential $B^{\pm}$ near the diagonal:
\begin{equation}\label{64}B^{\pm}_{reg}(x,x)=\Bigg(B^{\pm}(x,y)-\frac{v(x)v(y)}{(\int^y_x v)^2} \Bigg)\Bigg|_{y=x}. \end{equation}
The differential equations for $\tau^{\pm}$ with respect to the coordinates (\ref{22}) on $\mathcal{Q}_{g,m}[\mathbf{k}]$ are given by the following theorem:

\begin{theorem}
 Hodge and Prym tau-functions $\tau^{\pm}$ defined by (\ref{53}) and (\ref{58}) satisfy the following system of equations on $\mathcal{Q}_{g,m}[\mathbf{k}]$:
\begin{equation}\label{65} \frac{\partial \log \tau^{\pm}}{\partial A_j }=\frac{1}{4 \pi i}\oint_{b^-_j}\frac{B^{\pm}_{reg}}{v}, \quad \frac{\partial \log \tau^{\pm}}{\partial B_j }=-\frac{1}{4 \pi i}\oint_{a^-_j}\frac{B^{\pm}_{reg}}{v},\end{equation}
for j=$1,...,g^-$
\begin{equation}\label{66}\frac{\partial \log \tau^{\pm}}{\partial \, (2 \pi i \tilde{r}_k) }=-\frac{1}{4 \pi i}\oint_{\kappa^-_k}\Bigg(\frac{B^{\pm}_{reg}}{v}+\frac{1}{12\tilde{r}^2_k}v \Bigg),\end{equation}
for $k=1,...,n.$

\begin{equation}\label{67}\frac{\partial \log \tau^{\pm}}{\partial \, (2 \pi i r_k) }=-\frac{1}{4 \pi i}\oint_{\kappa^-_k}\frac{B^{\pm}_{reg}}{v} ,\end{equation}
for $k=n+1,...,m.$ Here $\tilde{r}_k$ are residues near simple poles, while $r_k$ are residues near higher order poles.

\end{theorem}

This is a natural extension of the results proven in \cite{Kokotov_2009}, \cite{korotkin2013tau} and \cite{Kalla_2014} to the case of higher order poles. While equations (\ref{66}) with respect to the residues near simple poles were recently present for $\log \tau^+ $ in \cite{bertola2021wkb}, the similar variations for $\log \tau^-$ could be obtained by a minor modification.

\begin{remark}
Local analysis shows that near $\tilde{z}^{(1)}_j$ and $\tilde{z}^{(2)}_j$  the expressions $\frac{B^{\pm}_{reg}}{v}$ gain simple poles. Then the addition appearing in (\ref{66}) regularize the integrand at the endpoints of the integration path. This issue does not emerge in case of higher order poles with $k_j \geq 2$, where both $\frac{B^{\pm}_{reg}}{v}$ gain a zero of order $k_j-2$ and equations (\ref{67}) are valid. 

\end{remark}

\section{Variational formulas on $\mathcal{M}_{SL(2)}[\mathbf{k}]$}
The assumption that moduli of the base curve $\mathcal{C}$ are kept fixed allows us to well-define the variations on $\mathcal{M}_{SL(2)}[\mathbf{k}]$ for any Abelian differential associated with spectral curve $\hat{\mathcal{C}}.$ For any fixed local chart $D$ and corresponding local coordinate $\xi$ on $\mathcal{C},$ we can lift $D$ to $\hat{\mathcal{C}}$ via $\pi^{-1}$ and use $\xi$ as a local coordinate on each connected component of $\pi^{-1}(D)$ outside the branch points. Then if $p_i$ is any coordinate on $\mathcal{M}_{SL(2)}[\mathbf{k}]$ from the list (\ref{9}), the variation of an Abelian differential $w=f(\xi) d \xi$ is defined by
\begin{equation}\label{68}\frac{d w}{d p_i}=\frac{d f(\xi)}{d p_i}d\xi, \end{equation}
assuming that the coordinate $\xi$ does not depend on $p_i.$ Such definition is clearly independent of the choice $\xi$ on $D.$
Before proving the following technical proposition, we introduce additional (meromorphic) Abelian differentials attributed to the double cover:
let $\Big \{\{\hat{w}^l_j\}^{k_j}_{l=2}\Big \}^m_{j=1}$  denote the second-kind differentials on $\hat{\mathcal{C}}$ with prescribed singular part normalized over $a$ - cycles of the homology basis (\ref{13}). That is,

\begin{equation}\label{69} \hat{w}^l_j(x)=\Big( \frac{1}{\chi_j^l}+O(1) \Big)d \chi_j, \quad \quad x \xrightarrow{} z_{j}^{(1)},\end{equation}
\begin{equation}\oint_{a_k} \hat{w}^l_j = \oint_{a^\mu_k}\hat{w}^l_j = \oint_{\tilde{a}_l}\hat{w}^l_j=0.\end{equation}

By $\{\hat{\eta}_j \}^m_{j=1}$ we denote  normalized over $a$ - cycles third-kind differentials with simple poles at $z_j^{(1)}$ and $x_1$ with residues $+1$ and $-1$, respectively.

Now put

\begin{equation}\label{70}{w}^{l-}_j:={\hat{w}^l_j-\mu^*\hat{w}^l_j},\end{equation}

\begin{equation}\label{71}{\eta}^{-}_j:={\hat{\eta}_j-\mu^*\hat{\eta}_j}.\end{equation}
We will call ${w}^{l-}_j$ the normalized \textit{Prym second-kind differential} and ${\eta}^{-}_j$ the normalized \textit{Prym third-kind differential}. The following lemma outlines the properties of defined objects
\begin{lemma}
The differentials ${w}^{l-}_j$ and ${\eta}^{-}_j$ have the following properties:

(i)
\begin{equation}\label{72}\mu^*{w}^{l-}_j=-{w}^{l-}_j,\end{equation}
(ii)
\begin{equation}\label{73}\mu^*{\eta}^{-}_j=-{\eta}^{-}_j,\end{equation}
(iii)
\begin{equation}\label{74}\oint_{a^{+}_k} {w}^{l-}_j=\oint_{a^{-}_k} {w}^{l-}_j=\oint_{b^{+}_k}{w}^{l-}_j=0,\end{equation}
(iv)
\begin{equation}\label{75}\oint_{a^{+}_k} {\eta}^{-}_j=\oint_{a^{-}_k} {\eta}^{-}_j=\oint_{b^{+}_k} {\eta}^{-}_j=0,\end{equation}
(v)
the differential ${w}^{l-}_j$ has the following singular parts:
\begin{equation}\label{76}{w}^{l-}_j(x)=\Bigg( \frac{1}{\chi_j^l}+O(1) \Bigg)d \chi_j, \quad \quad x \xrightarrow{} z_{j}^{(1)},\end{equation}
\begin{equation}\label{77}{w}^{l-}_j(x)=-\Bigg( \frac{1}{\chi_j^l}+O(1) \Bigg)d \chi_j, \quad \quad x \xrightarrow{} z_{j}^{(2)}\end{equation}
and it is holomorphic elsewhere,

(vi)
the differential ${\eta}^{-}_j$ is of third kind with simple poles at $z_j^{(1)}$ and $z_j^{(2)}$ with residues $+1$ and $-1$, respectively,

(vii)
\begin{equation}\label{78}\oint_{b^{-}_k} {w}^{l-}_j=2 \pi i \frac{{u}^{(l-2)-}_k(z^{(1)}_j)}{(l-1)!}, \qquad  \oint_{b^{-}_k} {\eta}^{-}_j= \pi i\int^{z^{(1)}_j}_{z^{(2)}_j} u^-_k,\end{equation}
where ${u}^{(l-2)-}_k(z^{(1)}_j)$ stands for $\frac{d^{l-2}}{d \chi^{l-2}_j}\Big(\frac{u^-_k}{d \chi_j} \Big)\Big|_{\chi_j=0}.$

\end{lemma}
\begin{proof}
(i)-(vi) follows from the definitions and the facts that $\mu(z_j^{(1)})=z_j^{(2)}, \ \mu(x_1)=x_1.$

To prove (vii) we apply the Riemann Bilinear Identity (later RBI) to ${w}^{l-}_j$ with $u^-_k$ $(\eta^-_j$ with $u^-_k)$ by integrating them over the cycles in $H_{-}(\hat{\mathcal{C}})$ in the following way:
consider f.e.
\begin{equation}\label{79}\oint_{b^{-}_k} {w}^{l-}_j=\sum^{g^-}_{j=1}\Big[\oint_{b^-_j}{w}^{l-}_j \oint_{a^-_j}u^-_k - \oint_{a^-_j}{w}^{l-}_j \oint_{b^-_j}u^-_k\Big].\end{equation}
We can assume that the boundary of the universal cover $\hat{\mathcal{C}}_0$ of $\hat{\mathcal{C}}$ is invariant under the involution $\mu.$ Then we can extend this sum by adding integrals over $H_+(\hat{\mathcal{C}}).$ The integrands are skew-symmetric with respect to involution, so their integrals over the cycles in $H_+(\hat{\mathcal{C}})$ give zero contribution. Then by the Stokes' theorem (\ref{79}) could be represented as the sum over residues near poles inside $\hat{\mathcal{C}}_0$:
\begin{equation}\label{80} \frac{1}{2} (2 \pi i)\underset{(z_j^{(1)}, z_j^{(2)})}{res}\Bigg({w}^{l-}_j\int^{x}_{p_0}u^-_{k}\Bigg), \end{equation}
where the factor $1/2$ is due to the intersection index $a^-_i \circ b^-_k=\frac{1}{2}\delta_{ik};$ $p_0$ is a reference point. This expression does not depend on the choice of the point $p_0$ since the difference between two choices is the sum of the residues of ${w}^{l-}_j$ which is clearly zero. For convenience, we put  $p_0=x_1.$ Then skew-symmetry of both differentials in (\ref{80}) implies that  the residues at $z_j^{(1)}$ and $z_j^{(2)}$ are equal. Their computation leads to the result. The second formula is proven by analogy, noticing that

\begin{equation}\label{81}2\int^{z^{(1)}_j}_{x_1} u^-_k=\int^{z^{(1)}_j}_{z^{(2)}_j} u^-_k. \end{equation}
\end{proof}

\begin{propos}
The following variational formulas with respect to the coordinates (\ref{9}) on $\mathcal{M}_{SL(2)}[\mathbf{k}]$ hold:
\begin{equation}\label{82}\frac{\partial v}{\partial A_{\alpha} }=u^-_{\alpha}, \end{equation}
\begin{equation}\label{83}\frac{\partial v}{\partial C^{l}_j }={w}^{l-}_j, \ l \geq 2, \end{equation}
\begin{equation}\label{84}\frac{\partial v}{\partial C^{1}_j }={\eta}^{-}_j. \end{equation}

\end{propos}
\begin{proof}
The proof follows the idea in \cite{bertola2019spaces}.
Consider an expansion of the Abelian differential $v$ near a branch point $x_i$ on $\hat{\mathcal{C}}.$ The coordinate could be taken as $(\xi-\xi_i)^{1/2}.$ While $\xi$ is moduli-independent, $\xi_i=\xi(\pi(x_i))$ changes when $\hat{\mathcal{C}}$ varies. Thus, the dependence of $\xi_i$ on moduli should be taken into account. $v$ has a double zero at each $x_i,$ then locally it can be written as follows:
\begin{equation}\label{85}v(\xi)=(\xi-\xi_j)(a_0+a_1(\xi-\xi_j)+...)d(\xi-\xi_j)^{1/2}=\frac{1}{2}(a_0(\xi-\xi_j)^{1/2}+a_1(\xi-\xi_j)^{3/2}+...)d\xi. \end{equation}
Performing the differentiation by the rule (\ref{68}) with respect to any coordinate $p$ we obtain.

\begin{equation}\label{86}\frac{\partial v}{\partial p}=\frac{1}{4} \Bigg(-\frac{a_0 (\xi_j)'_p}{(\xi-\xi_j)^{1/2}}+O(1) \Bigg)d\xi=-\frac{1}{2}\Bigg(a_0 (\xi_j)'_p+O(1) \Bigg)d(\xi-\xi_j)^{1/2}.\end{equation}
From this formula it is clear that $\frac{\partial v}{\partial p}$ is holomorphic at the branch points.

The differential $\frac{\partial v}{\partial A_{\alpha}}$ also holomorphic at all  poles $z_j$ since the singular parts of $v$ do not depend on the moduli $A_{\alpha}.$ Moreover, all periods of $\frac{\partial v}{\partial A_{\alpha}}$ vanish except for the period over $a^-_{\alpha},$ which is 1. Take the difference $w:=\frac{\partial v}{\partial A_{\alpha}}-u^-_{\alpha}.$ The differential $w$ is holomorphic and its $a^-$ and $a^+$ periods vanish by construction. Thus, we have that $w \equiv 0,$ and so (\ref{82}) holds. 

Consider $\frac{\partial v}{\partial C^{l}_j }.$ Its singular part coincides with the one of ${w}^{l-}_j.$ Also its $a^-$-periods vanish since $ C^{k}_j$ are independent of $\{A_{\alpha}\}^{g^-}_{\alpha=1}.$ Similarly to the previous argument we obtain (\ref{83}).

Finally, $\frac{\partial v}{\partial C^{1}_j }$ equals to the Prym third-kind differential ${\eta}^{-}_j$ again due to the coincidence of singular parts.
\end{proof}

\subsection{Variations of Prym matrix}
In this section we discuss variations of the Prym matrix $\Pi$ on the spaces of $SL(2)$ spectral covers. While the derivatives with respect to the $a^-$- periods reproduce the Prym version of Donagi-Markman cubic \cite{Donagi_1996}, variations with respect to residues and KP-times extend this result to meromorphic case and involve Prym meromorphic differentials.
\begin{theorem}
The variations of the Prym matrix $\Pi$ on the space $\mathcal{M}_{SL(2)}[\mathbf{k}]$ with respect to the coordinates (\ref{9}) take the following form:

\begin{equation}\label{87}\frac{\partial \Pi_{\alpha \beta}}{\partial A_{\gamma}}=-\pi i \sum^r_{i=1}\underset{x_i}{res}\Bigg(\frac{ {u}^-_{\alpha}  {u}^-_{\beta}  {u}^-_{\gamma}}{d{\xi} \: d( v/d{\xi})}\Bigg),\end{equation} 

\begin{equation}\label{88}\frac{\partial \Pi_{\alpha \beta}}{\partial C^{1}_j}=-\pi i \sum^r_{i=1}\underset{x_i}{res}\Bigg(\frac{  {u}^-_{\alpha}  {u}^-_{\beta} {\eta}^-_j }{d{\xi} \: d( v/d{\xi})}\Bigg),\end{equation} 

\begin{equation}\label{89}\frac{\partial \Pi_{\alpha \beta}}{\partial  C^{l}_j}=-\pi i \sum^r_{i=1}\underset{x_i}{res}\Bigg(\frac{ {u}^-_{\alpha}  {u}^-_{\beta}  {w}^{l-}_j}{d{\xi} \: d( v/d{\xi})}\Bigg),\end{equation} 
where $r=4g-4+2\sum^m_{j=1}k_j$ is a number of branch points (zeroes of $v$). $\xi$ denotes a local coordinate on $\mathcal{C}$ near a branch point $x_i.$ The above formulas do not depend on the choice of $\xi.$ 
\end{theorem}
\begin{proof}
Let us proof (\ref{87}). On the subspace $\mathcal{M}_{SL(2)}[\mathbf{k}] \subset \mathcal{Q}_{g,m}[\mathbf{k}]$ the coordinates $\{B_i\}$ become dependent functions of $\{A_i \}.$ Thus, we compute the derivative of Prym matrix applying the chain rule as follows:
\begin{equation}\label{90}\frac{d \Pi_{\alpha \beta}}{d A_{\gamma}}=\frac{\partial \Pi_{\alpha \beta}}{\partial A_{\gamma}}\Big|_{\{B_i\}^{g^-}_{i=1}=const}+\sum^{g^-}_{i=1}\frac{\partial \Pi_{\alpha \beta}}{\partial B_i} \frac{\partial B_i}{\partial A_{\gamma}}.\end{equation}
Using the variational formulas (\ref{40}) and (\ref{82}) we further rewrite this expression as
\begin{equation}\label{91}\frac{1}{2}\sum^{g^-}_{i=1} \Bigg( -\oint_{2b^-_i} \frac{{u}^-_{\alpha}  {u}^-_{\beta}}{v} \oint_{a^-_i} u^-_{\gamma} + \oint_{2a^-_i} \frac{{u}^-_{\alpha}  {u}^-_{\beta}}{v} \oint_{b^-_i} u^-_{\gamma} \Bigg).\end{equation}
Similarly to Lemma 1, this sum could be represented as a sum of residues inside the fundamental polygon $\hat{\mathcal{C}}_0$ of $\hat{\mathcal{C}}:$ 

\begin{equation}\label{92}
-\pi i\sum^r_{i=1} \underset{x_i}{res}\Bigg(\frac{ {u}^-_{\alpha}  {u}^-_{\beta}}{ v}\int^{x}_{p_0}u^-_{\gamma}\Bigg),
\end{equation}
where $p_0$ is a reference point. To simplify the residues at first notice that
\begin{equation}\label{93} \underset{x_i}{res}\Bigg(\frac{ {u}^-_{\alpha}  {u}^-_{\beta}}{ v}\int^{x}_{p_0}u^-_{\gamma}\Bigg)= \underset{x_i}{res}\Bigg(\frac{ {u}^-_{\alpha}  {u}^-_{\beta}}{ v}\int^{x}_{x_i}u^-_{\gamma}\Bigg),\end{equation}
(taking the difference of the above expressions it equals $\underset{x_i}{res}\Big(\frac{ {u}^-_{\alpha}  {u}^-_{\beta}}{ v}\Big)$ up to an explicit constant $C_i,$ then the skew-symmetry of the differential $\frac{ {u}^-_{\alpha}  {u}^-_{\beta}}{ v}$ implies it has vanishing residue).
We introduce a local coordinate $\hat{\xi}$ near a branch point on $\hat{\mathcal{C}}.$ $\hat{\xi}$ can be chosen such that \mbox{$\hat{\xi}(x^\mu)=-\hat{\xi}(x).$} 
Then for $u^-, v$ being skew-symmetric it implies that they have expansion by even powers of $\hat{\xi}$. Moreover, $v$ has double zeroes at branch points. Then one has:
\begin{equation}\label{94} {u}^-_{\alpha}=((u_0)_{\alpha}+O(\hat{\xi}^2))d\hat{\xi},\end{equation}
\begin{equation}\label{95} v=\hat{\xi}^2(v_0+O(\hat{\xi}^2))d\hat{\xi}.\end{equation}
And
\begin{equation}\label{96}\frac{ {u}^-_{\alpha}  {u}^-_{\beta}}{ v}= \Big(\frac{(u_0)_{\alpha}(u_0)_{\beta}}{v_0}\frac{1}{\hat{\xi}^2}+O(1) \Big)d\hat{\xi},\end{equation}
\begin{equation}\label{97}\int^{x}_{x_i}u^-_{\gamma}= \Big( (u_0)_{\gamma}\hat{\xi}+O(\hat{\xi}^3)\Big).\end{equation}
Therefore, the sum (\ref{92}) becomes
\begin{equation}\label{98}-\pi i \sum^r_{i=1}\frac { (u_0)_{\alpha}(u_0)_{\beta}(u_0)_{{\gamma}}}{v_0}, \end{equation}
which could be rewritten in invariant form (\ref{87})                (there $\xi-\xi_i=\hat{\xi}^2,$ and the expression $d \xi d(v/d\xi)$ is proportional exactly to $v_0 \hat{\xi} (d\hat{\xi})^2$).

The formula (\ref{88}) could be proven in a similar way. Applying the chain rule, we write:
\begin{equation}\label{99}\frac{d \Pi_{\alpha \beta}}{d C^1_j}=\frac{\partial \Pi_{\alpha \beta}}{\partial C^1_j} \Big|_{{\{A_i, B_i\}^{g^-}_{i=1}=const}}+\sum^{g^-}_{i=1} \Bigg( \frac{\partial \Pi_{\alpha \beta}}{\partial A_i} \frac{\partial A_i}{\partial C^1_j}+ \frac{\partial \Pi_{\alpha \beta}}{\partial B_i} \frac{\partial B_i}{\partial C^1_j}\Bigg).\end{equation}
According to (\ref{6}) and  (\ref{40})
\begin{equation}\label{100}\frac{\partial \Pi_{\alpha \beta}}{\partial C^1_j}=(2\pi i) \ \frac{1}{2} \int_{2\kappa^-_j} \frac{{u}^-_{\alpha}  {u}^-_{\beta}}{v}=\pi i \int^{z_j^{(2)}}_{z_j^{(1)}} \frac{{u}^-_{\alpha}  {u}^-_{\beta}}{v}. \end{equation}
Then using variation formulae (\ref{40}) and (\ref{84})
\begin{equation}\label{101}\frac{d \Pi_{\alpha \beta}}{d C^1_j}=\pi i \int^{z_j^{(2)}}_{z_j^{(1)}} \frac{{u}^-_{\alpha}  {u}^-_{\beta}}{v}+\frac{1}{2}\sum^{g^-}_{i=1} \Bigg( -\oint_{2b^-_i} \frac{{u}^-_{\alpha}  {u}^-_{\beta}}{v} \oint_{a^-_i} {\eta}^-_j + \oint_{2a^-_i} \frac{{u}^-_{\alpha}  {u}^-_{\beta}}{v} \oint_{b^-_i} {\eta}^-_j \Bigg).\end{equation}
With the help of the RBI we obtain the sum over residues at the branch points plus the residues at $\{z_j^{(1)}, z_j^{(2)} \}$ which conveniently cancel with the first integral.

Finally, to prove (\ref{89}), we use variations (\ref{40}) with (\ref{83}) to write
\begin{equation}\label{102}\frac{d \Pi_{\alpha \beta}}{d C^l_j}=\sum^{g^-}_{i=1} \Bigg( \frac{\partial \Pi_{\alpha \beta}}{\partial A_i} \frac{\partial A_i}{\partial C^l_j}+ \frac{\partial \Pi_{\alpha \beta}}{\partial B_i} \frac{\partial B_i}{\partial C^l_j}\Bigg)=\end{equation}
\begin{equation}\label{103}=\frac{1}{2}\sum^{g^-}_{i=1} \Bigg( -\oint_{2b^-_i} \frac{{u}^-_{\alpha}  {u}^-_{\beta}}{v} \oint_{a^-_i} {w}^{l-}_j + \oint_{2a^-_i} \frac{{u}^-_{\alpha}  {u}^-_{\beta}}{v} \oint_{b^-_i} {w}^{l-}_j \Bigg)\end{equation}
which due to the RBI is equal to the required expression  (notice that the residues at the poles $\{z_j^{(1)}, z_j^{(2)} \}$ of ${w}^{l-}_j$ vanish).
\end{proof}
\begin{remark}
By assumption the base curve $\mathcal{C}$ is kept fixed, so variations of its Period matrix $\Omega$ on $\mathcal{M}_{SL(2)}[\mathbf{k}]$ must be zero. Consider, f.e. $\frac{\partial \Omega_{\alpha \beta}}{\partial A_{\gamma}}.$ Then, similarly to the derivation of (\ref{87}), applying variational formulas (\ref{40}) it equals to
\begin{equation}\label{104}\frac{\partial \Omega_{\alpha \beta}}{\partial A_{\gamma}}=-\pi i\sum^n_{i=1} \underset{x_i}{res}\Bigg(\frac{ {u}^+_{\alpha}  {u}^+_{\beta}}{ v}\int^{x}_{p_0}u^-_{\gamma}\Bigg).\end{equation} 
The local analysis shows that each differential $u^+$ gains simple zeros at branch points $x_i$ when being lifted from the base curve $\mathcal{C}.$ Thus, the residues over $x_i$ vanish.
\end{remark}
The difference between the dimensions of $\mathcal{Q}_{g,m}[\mathbf{k}]$ and moduli space of curves $\mathcal{M}_{g,m}(\mathcal{C})$
\begin{equation}\label{105}dim(\mathcal{Q}_{g,m}[\mathbf{k}])-dim(\mathcal{M}_{g,m}(\mathcal{C})) =(6g-6+m+2\sum^m_{j=1}k_j)-(3g-3+m)=3g-3+2\sum^m_{j=1}k_j\end{equation}
implies the existence of $3g-3+2\sum^m_{j=1}k_j$ linearly independent fields on $\mathcal{Q}_{g,m}[\mathbf{k}]$ that preserve a complex structure of $\mathcal{C}$ and positions of poles. It is instructive to find these fields.

\begin{propos}

The following vector fields defined on $\mathcal{Q}_{g,m}[\mathbf{k}]$
 \begin{equation}\label{106}V_{A_\gamma}=\frac{\partial}{\partial {A_\gamma}} + \sum^{g^-}_{i=1} \Pi_{\gamma i} \frac{\partial}{\partial {B_i}}, \qquad \gamma=1,...,g^-, \end{equation}

\begin{equation}\label{107}V_{C^1_j}=\frac{\partial}{\partial  (2 \pi r_j)} +  \pi i \sum^{g^-}_{i=1} \Bigg(\int^{z^{(1)}_j}_{z^{(2)}_j} u^-_i \Bigg) \frac{\partial}{\partial {B_i}}, \qquad  j=1,...,m,\end{equation}

\begin{equation}\label{108}V_{C^l_j}= 2 \pi i \sum^{g^-}_{i=1}\Bigg(  \frac{{u}^{(l-2)-}_i(z^{(1)}_j)}{(l-1)!} \Bigg)  \frac{\partial}{\partial  B_i}, \qquad  j=1,...,m, \quad l=2,...,k_j.\end{equation}

preserve the moduli of $\mathcal{C}$ and positions of poles.
\end{propos}
\begin{proof}
   Consider a perturbation of the original quadratic differential $Q^{\epsilon}=Q+\epsilon\tilde{Q},$ where $\tilde{Q}$ is an arbitrary quadratic differential on $\mathcal{C}$ with simple zeroes and  poles at $z_j$ of even order no greater than $2k_j$. The differential $Q^{\epsilon}$ also has simple zeroes for $\epsilon$ small enough, is defined on the same Riemann surface and has the same set of poles as $Q.$ Thus, the vector field $\frac{d}{d \epsilon}$ does not change the complex structure of $\mathcal{C}$ and positions of poles. Expressing this vector field via the coordinates (\ref{22}) one has:

\begin{equation}
\frac{d}{d \epsilon}=\sum^m_{j=1}\Bigg(\oint_{t^-_j}\frac{\tilde{Q}}{2v}\Bigg)\frac{\partial} {\partial (2\pi i r_j)} +\sum^{g^-}_{i=1}\Bigg[\Bigg(\oint_{a^-_i}\frac{\tilde{Q}}{2v}\Bigg)\frac{\partial} {\partial A_i}+\Bigg(\oint_{b^-_i}\frac{\tilde{Q}}{2v}\Bigg)\frac{\partial} {\partial B_i} \Bigg].
\end{equation}
Then taking $\tilde{Q}$ to be equal either to $2v u_{\gamma}, 2v {\eta}^{-}_j$ or $2v{w}^{l-}_j$ and applying formulas (\ref{78}) we obtain the result.
Notice that (\ref{106}-\ref{108}) are exactly the expressions that appear in the proof of Theorem 5 when performing the chain rule. The fields are independent since on the submanifold $\mathcal{M}_{SL(2)}[\mathbf{k}]$ they are equal to the derivatives over independent coordinates.

\end{proof}
\begin{remark}
In the holomorphic case, when $m=0,$ one has the isomorphism between holomorphic quadratic differentials $\mathcal{Q}_{g}$ and the cotangent bundle of the moduli space of curves $T^* \mathcal{M}_g.$ It follows from the proposition that $V_{A_\gamma}, \gamma=1,...,3g-3$ act trivially on the coordinates $\{q_i\}^{3g-3}_{i=1}$ of the canonical Darboux coordinate set $\{p_i, q_i\}^{3g-3}_{i=1}$ on $T^* \mathcal{M}_g$. Thus, $V_{A_\gamma}$ span the vertical bundle $V_*(T^* \mathcal{M}).$ 
\end{remark}

\subsection{Variations of Prym differentials}
The derivatives of differentials depending on the point (points) on $\hat{\mathcal{C}}$ should be treated in a greater accuracy, since the latter is deforming. Similarly to (\ref{68}) we define on $\mathcal{M}_{SL(2)}[\mathbf{k}]$ derivatives for any coordinate $p_i$ in the list (\ref{9}) as follows
\begin{equation}\label{110}\frac{\partial  B^-(x,y)}{\partial p_i}=\frac{\partial}{\partial p_i}\Bigg(\frac{B^-(x,y)}{d \xi(x) d \xi(y)}\Bigg)d \xi(x) d \xi(y), \end{equation}
where $\xi$ is any fixed coordinate on $\mathcal{C}$ lifted to the covering surface via $\pi^{-1}.$ Notice that on the space $\mathcal{Q}_{g,m}[\mathbf{k}]$ the differentiation is performed according to the rule (\ref{42}), when the flat coordinates $z(x)$ and $z(y)$ are kept fixed. Restricting the variational formulas from $\mathcal{Q}_{g,m}[\mathbf{k}]$ onto its subspace $\mathcal{M}_{SL(2)}[\mathbf{k}]$ we obtain the following:
\begin{theorem}
The variations of the Prym bidifferential $B^-(x,y)$ on the space $\mathcal{M}_{SL(2)}[\mathbf{k}]$ with respect to the coordinates (\ref{9}) take the following form:
\begin{equation}\label{111}\frac{\partial B^-(x,y)}{\partial A_{\gamma}}=-\frac{1}{2} \sum^r_{i=1}\underset{x_i}{res}\Bigg(\frac{{u}^-_{\gamma}(t) B^-(x,t)  B^-(t,y)  }{d{\xi} \: d( v/d{\xi})}\Bigg),\end{equation} 

\begin{equation}\label{112}\frac{\partial B^-(x,y)}{\partial C^{1}_j}=-\frac{1}{2} \sum^r_{i=1}\underset{x_i}{res}\Bigg(\frac{{\eta}^{-}_j(t) B^-(x,t)  B^-(t,y)  }{d{\xi} \: d( v/d{\xi})}\Bigg),\end{equation} 

\begin{equation}\label{113}\frac{\partial B^-(x,y)}{\partial  C^{l}_j}=-\frac{1}{2} \sum^r_{i=1}\underset{x_i}{res}\Bigg(\frac{{w}^{l-}_j(t) B^-(x,t)  B^-(t,y)  }{d{\xi} \: d( v/d{\xi})}\Bigg),\end{equation}
$\xi$ denotes a local coordinate on $\mathcal{C}$ near a branch point $x_i.$ The above formulas do not depend on the choice of $\xi.$ 
\end{theorem}
\begin{proof}
Let us prove (\ref{111}). Denote by $b^-(x,y):=\frac{B^-(x,y)}{v(x)v(y)}.$ Then one has

\begin{equation}\label{114}\frac{\partial B^-(x,y)}{\partial A_{\gamma}}\Big|_{\xi(x),\xi(y)}=\frac{\partial [b^-(x,y)v(x)v(y)]}{\partial A_{\gamma}}\Big|_{\xi(x),\xi(y)}=\end{equation}
\begin{equation}\label{115}=\frac{\partial b^-(x,y)}{\partial A_{\gamma}}\Big|_{\xi(x),\xi(y)}v(x)v(y)+b^-(x,y)v(y)\frac{\partial v(x)}{\partial A_{\gamma}}\Big|_{\xi(x)}+b^-(x,y)v(x)\frac{\partial v(y)}{\partial A_{\gamma}}\Big|_{\xi(y)} =\end{equation}
\begin{equation}\label{116}\stackrel{(\ref{82})}{=}\frac{\partial b^-(x,y)}{\partial A_{\gamma}}\Big|_{z(x),z(y)}v(x)v(y)+\frac{\partial b^-(x,y)}{\partial z(x)}\frac{\partial z(x)}{\partial A_{\gamma}}\Big|_{\xi(x),\xi(y)}v(x)v(y)+\frac{\partial b^-(x,y)}{\partial z(y)}\frac{\partial z(y)}{\partial A_{\gamma}}\Big|_{\xi(x),\xi(y)}v(x)v(y)+\end{equation} $$+b^-(x,y)v(y)u^-_{\gamma}(x)+b^-(x,y)v(x)u^-_{\gamma}(y)= $$
\begin{equation} \label{eqv}
\stackrel{(\ref{82})}{=}\frac{\partial b^-(x,y)}{\partial A_{\gamma}}\Big|_{z(x),z(y)}dz(x)dz(y)+ \Bigg[\Big(b^-(x,y)\Big)'_{z(x)}\int^x_{x_1} u^-_{\gamma} +\Big(b^-(x,y)\Big)'_{y(x)}\int^y_{x_1} u^-_{\gamma} \Bigg] dz(x)dz(y)+    
\end{equation}
$$+b^-(x,y)u^-_{\gamma}(x)dz(y)+b^-(x,y)u^-_{\gamma}(y)dz(x).$$

To compute the term $\frac{\partial b^-(x,y)}{\partial A_{\gamma}}\Big|_{z(x),z(y)}$ we, similarly to (\ref{90}), apply the chain rule and then variational formulas (\ref{42}), (\ref{82}) along with the RBI to obtain

\begin{equation}\label{118}\frac{\partial b^-(x,y)}{\partial A_{\gamma}}\Big|_{z(x),z(y)}= -\sum_{t \in int(\hat{C})}\frac{1}{2}\underset{t}{res}\Bigg(  {b}^-(x,t)  {b}^-(t,y)  v(t)\int^{t}_{p_0}u^-_{\gamma} \Bigg). \end{equation}
To evaluate the residues, introduce the differentials $V(t)= {b}^-(x,t)  {b}^-(t,y)  v(t)$ and $W(t)=u^-_{\gamma}(t)$.
$W(t)$ is holomorphic, while  $V(t),$ in addition to second order poles at $\{x_i \}^{r}_{i=1}$ has four second order poles at $t=x, \ x^\mu$ and at $t=y, \ y^\mu.$ As in the proof of Lemma 1, we can put  $p_0=x_1.$ Then residues at the pairs $(x, x^\mu)$ and $(y, y^\mu)$ are the same due to $V(t), W(t)$ being skew-symmetric under the involution and, thus, their contribution to the sum double. Using the expansion of the Prym differential (\ref{33}) the residue at $t=x$ equals
\begin{equation}\label{119}\Big( {b}^-(x,y)\int^{x}_{x_1}u^-_{\gamma} \Big)'_{z(x)}. \end{equation}
The residue at $t=y$ is
\begin{equation}\label{120}\Big( {b}^-(x,y)\int^{y}_{x_1}u^-_{\gamma} \Big)'_{z(y)}. \end{equation}
Then it it easy to see that these terms, multiplied by $dz(x)dz(y),$ cancel the last four terms of the sum (\ref{eqv}). The evaluation of residues at $x_i$, similarly to the proof of (\ref{87}), leads to the result. (\ref{112}) and (\ref{113}) are obtained by analogy.
\end{proof}
Integrating above formulas over the cycles $b^-$ an using the simple fact that
\begin{equation}\label{121}\oint_{b^{-}_{\alpha}} B^-(\cdot,y)= 2 \pi i\  u^{-}_{\alpha} (y) \end{equation} we derive variational formulas for Prym normalized differentials:
\begin{theorem}
The variations of normalized Prym differentials $u^-_{\alpha}$ on the space $\mathcal{M}_{SL(2)}[\mathbf{k}]$ with respect to the coordinates (\ref{9}) take the following form:
\begin{equation}\label{122}\frac{\partial u^-_{\alpha}(x)}{\partial A_{\gamma}}=-\frac{1}{2} \sum^r_{i=1}\underset{x_i}{res}\Bigg(\frac{{u}^-_{\gamma}(t) B^-(x,t)  u^-_{\alpha}(t)  }{d{\xi} \: d( v/d{\xi})}\Bigg),\end{equation} 

\begin{equation}\label{123}\frac{\partial u^-_{\alpha}(x)}{\partial C^{1}_j}=-\frac{1}{2} \sum^r_{i=1}\underset{x_i}{res}\Bigg(\frac{\eta^-_j(t) B^-(x,t)  u^-_{\alpha}(t)  }{d{\xi} \: d( v/d{\xi})}\Bigg),\end{equation} 

\begin{equation}\label{124}\frac{\partial u^-_{\alpha}(x)}{\partial  C^{l}_j}=-\frac{1}{2} \sum^r_{i=1}\underset{x_i}{res}\Bigg(\frac{{w}^{l-}_j(t) B^-(x,t)  u^-_{\alpha}(t)  }{d{\xi} \: d( v/d{\xi})}\Bigg),\end{equation}
$\xi$ denotes a local coordinate on $\mathcal{C}$ near a branch point $x_i.$ The above formulas do not depend on the choice of $\xi.$ 
\end{theorem}
\begin{remark}
By analogy with Remark 5 one can show that the variations of $B^+(x,y)$ and 
$u^+_{\alpha}$ on $\mathcal{M}_{SL(2)}[\mathbf{k}]$ are zero. There is no surprise since these objects are the pullbacks from the base curve $\mathcal{C}$ which is assumed not to depend on moduli.
\end{remark}

\subsection{Variations of tau functions}
Initially Hodge and Prym tau-functions defined by (\ref{53}) and (\ref{58}) solve the system of differential equations on the space $\mathcal{Q}_{g,m}[\mathbf{k}]$ with a variable base. Having the base curve $\mathcal{C}$ fixed on $\mathcal{M}_{SL(2)}[\mathbf{k}]$ we define tau-functions on this subspace by a natural restriction of $\tau^{\pm}$ from $\mathcal{Q}_{g,m}[\mathbf{k}]$ assuming that the Period matrix $\Omega$ is constant.

Let $\{ \tilde{C}^1_j \}^n_{j=1}$ denote the local coordinates on $\mathcal{M}_{SL(2)}[\mathbf{k}]$ from the list (\ref{9}) corresponding to the residues near simple poles $\{\tilde{z}^{(1)}_k,\tilde{z}^{(2)}_k\}^n_{j=1}$ of the differential $v,$ whereas coordinates $\{ C^1_j \}^m_{j=n+1}$ are residues near the higher order poles.

\begin{theorem}
Prym tau-function $\tau^-$ satisfies the following system of differential equations on the space $\mathcal{M}_{SL(2)}[\mathbf{k}]$

\begin{equation}\label{125} \frac{\partial \log(\tau^-)}{\partial A_\gamma}= \frac{1}{2} \sum^r_{i=1}\underset{x_i}{res}\Bigg(\frac{{u}^-_{\gamma} \hat{B}_{reg}  }{d{\xi} \: d( v/d{\xi})}\Bigg)+\frac{11}{432}\sum^r_{i=1}\underset{x_i}{res}\Bigg(\frac{u^-_{\gamma}}{\int^x_{x_i}v} \Bigg)-\sum^n_{k=1}\frac{1}{48\tilde{r}_k}\int^{\tilde{z}^{(1)}_k}_{\tilde{z}^{(2)}_k}u^-_{\gamma},\end{equation} 

\begin{equation}\label{126}  \frac{\partial \log(\tau^-)}{\partial \tilde{C}^{1}_j}=\frac{1}{2} \sum^r_{i=1}\underset{x_i}{res}\Bigg(\frac{\eta^-_j \hat{B}_{reg}  }{d{\xi} \: d( v/d{\xi})}\Bigg)+\frac{11}{432}\sum^r_{i=1}\underset{x_i}{res}\Bigg(\frac{\eta^-_j}{\int^x_{x_i}v} \Bigg)-\sum^n_{k=1, k \neq j}\frac{1}{48\tilde{r}_k}\int^{\tilde{z}^{(1)}_k}_{\tilde{z}^{(2)}_k}\eta^-_j-
\end{equation}

$$-\frac{1}{48\tilde{r}_j}\int^{\tilde{z}^{(1)}_j}_{\tilde{z}^{(2)}_j}\Bigg(\frac{v}{\tilde{r}_j}-\eta^-_j \Bigg),$$ 
$\quad j=1,...,n,$
\begin{equation}\label{127}\frac{\partial \log(\tau^-)}{\partial C^{1}_j}=\frac{1}{2} \sum^r_{i=1}\underset{x_i}{res}\Bigg(\frac{\eta^-_j \hat{B}_{reg}  }{d{\xi} \: d( v/d{\xi})}\Bigg)+\frac{11}{432}\sum^r_{i=1}\underset{x_i}{res}\Bigg(\frac{\eta^-_j}{\int^x_{x_i}v} \Bigg)-\sum^n_{k=1}\frac{1}{48\tilde{r}_k}\int^{\tilde{z}^{(1)}_k}_{\tilde{z}^{(2)}_k}\eta^-_j, \end{equation}
$\quad j=n+1,...,m, $

\begin{equation}\label{128}\frac{\partial \log(\tau^-)}{\partial  C^{l}_j}=\frac{1}{2} \sum^r_{i=1}\underset{x_i}{res}\Bigg(\frac{w^{l-}_j \hat{B}_{reg}  }{d{\xi} \: d( v/d{\xi})}\Bigg)+\frac{11}{432}\sum^r_{i=1}\underset{x_i}{res}\Bigg(\frac{w^{l-}_j}{\int^x_{x_i}v} \Bigg)-\sum^n_{k=1}\frac{1}{48\tilde{r}_k}\int^{\tilde{z}^{(1)}_k}_{\tilde{z}^{(2)}_k}w^{l-}_j,\end{equation} 
$\quad j=n+1,...,m, \ \  l=2,...,k_j-1.$

\begin{equation}\label{129}\frac{\partial \log(\tau^-)}{\partial  C^{k_j}_j}=\frac{1}{2} \sum^r_{i=1}\underset{x_i}{res}\Bigg(\frac{w^{k_j-}_j \hat{B}_{reg}  }{d{\xi} \: d( v/d{\xi})}\Bigg)+\frac{11}{432}\sum^r_{i=1}\underset{x_i}{res}\Bigg(\frac{w^{k_j-}_j}{\int^x_{x_i}v} \Bigg)-\sum^n_{k=1}\frac{1}{48\tilde{r}_k}\int^{\tilde{z}^{(1)}_k}_{\tilde{z}^{(2)}_k}w^{k_j-}_j+ \end{equation}

$$
+\frac{1}{(k_j-1)} \frac{2k_j-k^2_j}{24 C^{k_j}_j},$$
$\quad j=n+1,...,m.$
\end{theorem}

\begin{proof}
In parallel to (\ref{87}), we apply the chain rule to the equations (\ref{65}) and use (\ref{82}) with the RBI to have
\begin{equation}\label{131}\frac{\partial \log(\tau^-)}{\partial A_{\gamma}}=\frac{1}{4}\sum_{t=\{x_i, \tilde{z}_k^{(1)}, \tilde{z}_k^{(2)} \} }\underset{t}{res} \Bigg(\frac{B^-_{reg}}{v} \int^x_{x_1}u^-_{\gamma}\Bigg). \end{equation}
Notice that in addition to the residues at branch points $(x_i)^r_{i=1}$ we have residues at simple poles of $v$ at $(\tilde{z}_k^{(1)}, \tilde{z}_k^{(2)})^{n}_{k=1}$. To compute the residue near $x_i$ we represent $B^{-}_{reg}$ as the difference of projective connections \cite{Kokotov_2009}:
\begin{equation}\label{132}B^{-}_{reg}=\frac{S^{-}_B-S_v}{6}.\end{equation} While $S^{-}_B$ is the Prym projective connection appearing in (\ref{33}), $S_v$ is Schwarzian projective connection defined by
\begin{equation}\label{133}S_v=\Big(\frac{v'}{v}\Big)' -\frac{1}{2}\Big(\frac{v'}{v}\Big)^2,\end{equation}
where $v'=(v/d\hat{\xi})'$ for any local coordinate $\hat{\xi}$ on $\hat{\mathcal{C}}.$ Recall that by (\ref{34}) we have $S^-_{B}(x)=\hat{S}_{{B}}(x)-6 \mu^*_y \hat{B}(x,y)|_{x=y}.$ In the neighborhood of $x_i$ we chose a local coordinate $\hat{\xi}$ such that $v=\hat{\xi}^2 d \hat{\xi}$. Then near $x_i$ we have
\begin{equation}\label{134}\frac{S_v}{6v}=-\frac{2}{3\hat{\xi}^4} d \hat{\xi}.\end{equation}
Moreover, we can choose $\hat{\xi}$ in such a way that $\hat{\xi}(\mu(x))=-\hat{\xi}(x)$. Therefore, near $x_i$ we also have
\begin{equation}\hat{B}(x, \mu(x)) = \Bigg[ \frac{1 }{(\hat{\xi}(x)-\hat{\xi}(\mu(x)))^2}+\frac{1}{6}\hat{S}_B(\hat{\xi}(x))+O\Big((\hat{\xi}(x)-\hat{\xi}(\mu(x)))^2 \Big) \Bigg] d \hat{\xi}(x) d \hat{\xi}(\mu(x))= \end{equation}
\begin{equation}=\Bigg[ -\frac{1}{4\hat{\xi}^2}-\frac{1}{6}\hat{S}_B(\hat{\xi})+O(\hat{\xi}^2)\Bigg](d\hat{\xi})^2, \end{equation}
so that
\begin{equation}\frac{S^-_B}{6v}=\Bigg[ \frac{1}{4\hat{\xi}^4}+\frac{\hat{S}_B(\hat{\xi})}{3 \hat{\xi}^2}+O(1)\Bigg]d\hat{\xi} \end{equation}
and
\begin{equation}\frac{B^-_{reg}}{v}=\Bigg[\frac{11}{12\xi^4}+\frac{\hat{S}_B(\hat{\xi})}{3 \hat{\xi}^2}+O(1) \Bigg] d\hat{\xi}. \end{equation}
By analogy with the computation of the residue near the second order pole in (\ref{92}) one has
\begin{equation}\label{135}\underset{x_i}{res} \Bigg(\frac{\hat{S}_B(\hat{\xi})d \hat{\xi}}{3\hat{\xi}^2} \int^x_{x_1}u^-_{\gamma}\Bigg)=\frac{1}{3}\underset{x_i}{res}\Bigg(\frac{{u}^-_{\gamma} \hat{S}_B  }{d{\xi} \: d( v/d{\xi})}\Bigg)=2\underset{x_i}{res}\Bigg(\frac{{u}^-_{\gamma} \hat{B}_{reg}  }{d{\xi} \: d( v/d{\xi})}\Bigg). \end{equation}
Also
\begin{equation}\label{136}\underset{x_i}{res} \Bigg(\frac{11d \hat{\xi}}{12\hat{\xi}^4} \int^x_{x_1}u^-_{\gamma}\Bigg)=\frac{11}{12}\frac{1}{3!}\Bigg( \frac{u^-_{\gamma}}{ d\hat{\xi}} \Bigg)''(x_i)=\frac{11}{108}\underset{x_i}{res}\Bigg(\frac{u^-_{\gamma}}{\int^x_{x_i}v} \Bigg). \end{equation}
To compute residues near simple poles $\tilde{z}_k$ we use the local coordinate $\zeta$ (\ref{46}) to write near $\tilde{z}^{(1)}_k:$
\begin{equation}\label{137} 
    \frac{1}{6}\frac{S^-_B-S_v}{v}=\frac{1}{6}\frac{S^-_B(\zeta)-\frac{1}{2\zeta^2}}{\frac{\tilde{r}_k}{\zeta}}d\zeta=\Bigg(-\frac{1}{12 \tilde{r}_k \zeta} +O(1)\Bigg) d\zeta. 
\end{equation}
Thus,
\begin{equation}\label{138}(\underset{z^{(1)}_k}{res}+\underset{z^{(2)}_k}{res} )\Bigg(\frac{B^-_{reg}}{v} \int^x_{x_1}u^-_{\gamma}\Bigg)=-\frac{1}{12\tilde{r}_k}\int^{\tilde{z}^{(1)}_k}_{\tilde{z}^{(2)}_k}u^-_{\gamma}\end{equation}
and the formula (\ref{125}) results.

To obtain (\ref{126}) we apply the chain rule with (\ref{65}), (\ref{66}) and (\ref{84}) to write
\begin{equation}\label{139}\frac{d \Pi_{\alpha \beta}}{d \tilde{C}^1_j}=-\frac{1}{4} \int^{z_j^{(2)}}_{z_j^{(1)}} \Bigg(\frac{B^{-}_{reg}}{v}+\frac{1}{12\tilde{r}^2_j}v \Bigg)+\frac{1}{8 \pi i}\sum^{g^-}_{i=1} \Bigg( \oint_{2b^-_i} \frac{B^{-}_{reg}}{v} \oint_{a^-_i} {\eta}^-_j - \oint_{2a^-_i} \frac{B^{-}_{reg}}{v} \oint_{b^-_i} {\eta}^-_j \Bigg).\end{equation}
Notice that in this case both differentials $({B^-_{reg}}/{v})$ and $\eta^-_j$ have simple poles at $(\tilde{z}_j^{(1)}, \tilde{z}_j^{(2)}).$ From (\ref{137}) it follows that in order to regularize $({B^-_{reg}}/{v})$ near these points we need to add $\frac{1}{12 \tilde{r}_j} \eta^-_j$. Then the sum could be rewritten as
\begin{equation}\label{140} \frac{1}{8 \pi i} \sum^{g^-}_{i=1} \Bigg[ \oint_{2b^-_i} \Bigg( \frac{B^{-}_{reg}}{v} +\frac{1}{12 \tilde{r}_j} \eta^-_j \Bigg) \oint_{a^-_i} {\eta}^-_j - \oint_{2a^-_i} \Bigg( \frac{B^{-}_{reg}}{v} +\frac{1}{12 \tilde{r}_j} \eta^-_j \Bigg) \oint_{b^-_i} {\eta}^-_j \Bigg], \end{equation}
which is due to the RBI equals

\begin{equation}\label{141}-\frac{1}{4}\sum_{t=\{x_i, \tilde{z}_k^{(1)}, \tilde{z}_k^{(2)} \} }\underset{t}{res} \Bigg[ \eta^-_j \int^x_{x_1}\Bigg( \frac{B^{-}_{reg}}{v} +\frac{1}{12 \tilde{r}_j} \eta^-_j \Bigg) \Bigg]\end{equation}
and the evaluation of residues provides the formula (\ref{126}).

For $\frac{\partial \log(\tau^-)}{\partial  C^{1}_j}, \ j=n+1,...,m$ we have, using (\ref{65}), (\ref{67}) and (\ref{84}),

\begin{equation}\label{142}\frac{d \Pi_{\alpha \beta}}{d C^1_j}=-\frac{1}{4} \int^{z_j^{(2)}}_{z_j^{(1)}} \Bigg(\frac{B^{-}_{reg}}{v} \Bigg)+\frac{1}{8 \pi i}\sum^{g^-}_{i=1} \Bigg( \oint_{2b^-_i} \frac{B^{-}_{reg}}{v} \oint_{a^-_i} {\eta}^-_j - \oint_{2a^-_i} \frac{B^{-}_{reg}}{v} \oint_{b^-_i} {\eta}^-_j \Bigg).\end{equation}
Here differential $\frac{B^-_{reg}}{v}$ is holomorphic at $z^{(1)}_j$ and $z^{(2)}_j$ where it gains a zero of order $k_j-2.$ Then no regularization needed and applying the RBI we obtain (\ref{127}).

Similarly to (\ref{125}) one derives (\ref{128}) and (\ref{129}) for $\frac{\partial\log(\tau^-)}{\partial  C^{l}_j}, \ l \geq 2.$ For $l=k_j$ extra term appears due to nontrivial coinciding residues near poles $z^{(1)}_j$ and $z^{(1)}_j$ of $w^{k_j-}_j$:
\begin{equation}\frac{1}{(k_j-1)!}\Bigg(\frac{S_v}{6v}\Bigg)^{(k_j-2)}\big(z^{(1)}_j \big),\end{equation}
where the derivative is taken in a local coordinate $\chi_j.$
Using the expansion $(\ref{6})$ of $v$ and the formula for $S_v$ from (\ref{133}) one derives
\begin{equation}\frac{1}{(k_j-1)!}\Bigg(\frac{S_v}{6v}\Bigg)^{(k_j-2)}\big(z^{(1)}_j \big)= \frac{1}{(k_j-1)} \frac{2k_j-k^2_j}{12 C^{k_j}_j},\end{equation}
which finalize the computation.
\end{proof}

\begin{theorem}
Hodge tau-function $\tau^+$ satisfies the following system of differential equations on the space $\mathcal{M}_{SL(2)}[\mathbf{k}]$

\begin{equation}\label{143} \frac{\partial \log(\tau^+)}{\partial A_\gamma}= \frac{5}{432}\sum^r_{i=1}\underset{x_i}{res}\Bigg(\frac{u^-_{\gamma}}{\int^x_{x_i}v} \Bigg)-\sum^n_{k=1}\frac{1}{48\tilde{r}_k}\int^{\tilde{z}^{(1)}_k}_{\tilde{z}^{(2)}_k}u^-_{\gamma},\end{equation} 

\begin{equation}\label{144}  \frac{\partial \log(\tau^+)}{\partial \tilde{C}^{1}_j}=\frac{5}{432}\sum^r_{i=1}\underset{x_i}{res}\Bigg(\frac{\eta^-_j}{\int^x_{x_i}v} \Bigg)-\sum^n_{k=1, k \neq j}\frac{1}{48\tilde{r}_k}\int^{\tilde{z}^{(1)}_k}_{\tilde{z}^{(2)}_k}\eta^-_j-\frac{1}{48\tilde{r}_j}\int^{\tilde{z}^{(1)}_j}_{\tilde{z}^{(2)}_j}\Bigg(\frac{v}{\tilde{r}_j}-\eta^-_j \Bigg), \end{equation} 
$\quad j=1,...,n,$
\begin{equation}\label{145}\frac{\partial \log(\tau^+)}{\partial C^{1}_j}=\frac{5}{432}\sum^r_{i=1}\underset{x_i}{res}\Bigg(\frac{\eta^-_j}{\int^x_{x_i}v} \Bigg)-\sum^n_{k=1}\frac{1}{48\tilde{r}_k}\int^{\tilde{z}^{(1)}_k}_{\tilde{z}^{(2)}_k}\eta^-_j, \end{equation}
$\quad j=n+1,...,m, $

\begin{equation}\label{146}\frac{\partial \log(\tau^+)}{\partial  C^{l}_j}=\frac{5}{432}\sum^r_{i=1}\underset{x_i}{res}\Bigg(\frac{w^{l-}_j}{\int^x_{x_i}v} \Bigg)-\sum^n_{k=1}\frac{1}{48\tilde{r}_k}\int^{\tilde{z}^{(1)}_k}_{\tilde{z}^{(2)}_k}w^{l-}_j,\end{equation} 
$\quad j=n+1,...,m, \ \  l=2,...,k_j-1.$

\begin{equation}\label{147}\frac{\partial \log(\tau^+)}{\partial  C^{k_j}_j}=\frac{5}{432}\sum^r_{i=1}\underset{x_i}{res}\Bigg(\frac{w^{k_j-}_j}{\int^x_{x_i}v} \Bigg)-\sum^n_{k=1}\frac{1}{48\tilde{r}_k}\int^{\tilde{z}^{(1)}_k}_{\tilde{z}^{(2)}_k}w^{k_j-}_j+\frac{1}{(k_j-1)} \frac{2k_j-k^2_j}{24 C^{k_j}_j},\end{equation} 
$\quad j=n+1,...,m.\ \ $

\begin{proof}
The calculation could be performed in like manner to the previous theorem. The only difference is that the term $\frac{{B^{+}_{reg}}}{v}$ appearing in the variational formulas (\ref{65}-\ref{67}) for $\log \tau^+$ expands in the local coordinate $\hat{\xi}$ near $x_i$ as
\begin{equation}\frac{{B^{+}_{reg}}}{v}=\frac{S^{+}_B-S_v}{6v}=\Bigg[\frac{5}{12\xi^4}+O(1) \Bigg] d\hat{\xi},\end{equation}
which leads to the appearance of residues over branch points in the above formulas. All remaining computations of residues near $(\tilde{z}_j^{(1)}, \tilde{z}_j^{(2)})$ are similar to Theorem 8.

\end{proof}
\end{theorem}
\begin{remark}
It follows from (\ref{58}) that
$\log \hat{\tau}=\log \tau^+ + \log \tau^- .$ Then the variations of the tau function $\log \hat{\tau}$ could be obtained as the sum of the corresponding variational formulas from Theorems 8 and 9.
\end{remark}

\subsection{Higher order variations}
Here we briefly discuss higher order variations associated with the Prym matrix $\Pi_{\alpha \beta}$ in holomorphic case $(m=0).$ Denote the corresponding space by $\mathcal{M}_{SL(2)}.$ It is known that the period matrix of the spectral $GL(n)$ cover is given by second derivatives of a single function $\mathcal{F}$ called the prepotential. The same result holds for $\Pi_{\alpha \beta}$ in the $SL(2)$ case:

\begin{equation}\mathcal{F}=\frac{1}{2}\sum^{g^-}_{\gamma=1}A_{\gamma} B_{\gamma},\end{equation}
such that

\begin{equation} \label{prym}
    \Pi_{\alpha \beta}=\frac{\partial^2 \mathcal{F}}{\partial A_{\alpha} \partial A_{\beta}}.\end{equation}
The proof is in parallel with that outlined in \cite{bertola2019spaces}.

\begin{remark}
Prepotential $\mathcal{F}$ is  known to be a generating function between homological and canonical coordinates on symplectic space $\mathcal{Q}_g$ of holomorphic quadratic differentials \cite{Bertola_2017}.
\end{remark}

Then the cubic (\ref{87}) is given by a third derivate with respect to $A_\alpha$-periods on $\mathcal{M}_{SL(2)}:$

\begin{equation}\frac{\partial^3 \mathcal{F}}{\partial A_{\alpha}\partial A_{\beta}\partial A_{\gamma}},\end{equation}
or, equivalently, as a third Lie-derivative along the fields (\ref{106}) on $\mathcal{Q}_{g}:$

\begin{equation}V_{A_\alpha} V_{A_\beta} V_{A_\gamma} (\mathcal{F}).\end{equation}

We conclude with the formula for the second variation of the Prym matrix $\frac{\partial^2 \Pi_{\alpha \beta}}{\partial A_\delta \partial A_\gamma}$ . Due to (\ref{prym}) the resulting expression must be symmetric with respect to all 4 indices. It could be computed by differentiation of the formula (\ref{87}) and using variations of Prym differentials (\ref{122}). The computation could be performed following the Proposition 5.1 in \cite{bertola2019spaces} with a small alteration.

\begin{propos}
Second derivative of the Prym matrix $\Pi_{\alpha \beta}$ on $\mathcal{M}_{SL(2)}$ is given by the following expression:

\begin{equation}\label{148}\begin{aligned}&\frac{1}{2\pi i} \frac{\partial^2 \Pi_{\alpha \beta}}{\partial A_\delta \partial A_\gamma}=\frac{1}{16}\sum_{x_a \neq x_b} \Bigg \{B^-(x_a,x_b) \frac{u^-_\delta(x_a) u^-_\gamma(x_a) u^-_\alpha(x_b) u^-_\beta (x_b)+\textit{cycl of } (\alpha,\beta,\gamma)}{y'(x_a)y'(x_b)} \Bigg\}+\\
&+\frac{1}{16}\sum_{x_a}\Bigg \{\Bigg(\frac{6 \hat{B}_{reg}}{y'^2}-\frac{y'''}{y'^3} \Bigg){u}^-_\alpha{u}^-_\beta {u}^-_\gamma{u}^-_\delta(x_a) +\frac{1}{y'^2}(({u}^-_\alpha)''{u}^-_\beta {u}^-_\gamma{u}^-_\delta(x_a)+\textit{cycl of } (\alpha,\beta,\gamma,\delta)) \Bigg \}, \end{aligned}\end{equation}
where $y=\frac{v}{d\xi}.$ Values and derivatives at brach points $x_a$ are computed in a local coordinate $\hat{\xi}$. This formula does not depend on the choice of local coordinates $\xi, \hat{\xi}$ on $\mathcal{C}$ and $\hat{\mathcal{C}}$, provided $\hat{\xi}^2=\xi.$
\end{propos}

\textbf{Open problems.}
An immediate generalization of the above results will be the derivation of variational formulas on the spaces of meromorphic N-fold covers defined by $v^N=Q,$ where $Q$ is a meromorphic N-differential. Such covers possess a natural $\Z_N-$ symmetry which also induces splitting of the homology group by eigenspaces of the discrete group action. The formulas for the variable base were derived in \cite{korotkin2017tau}, \cite{korotkin2020bergman}. To treat the fixed base case an appropriate version of the chain rule and Riemann Bilinear Identity should be applied. Notice that this case is highly non-generic, since all ramification points are of order $N.$

\nocite{*}
\bibliographystyle{plain}

\end{document}